\documentclass[
	10pt, 
	two column, 
	twoside
]{IEEEtran}
\usepackage{amssymb,amsmath,mathrsfs}
	
\usepackage[ 
	lined,
	boxed,
	commentsnumbered, 
	ruled
]{algorithm2e}
\usepackage{stfloats}
\usepackage{subfigure}
\usepackage{graphicx,booktabs,multirow}
\usepackage[noadjust]{cite}
\usepackage{dsfont}
\usepackage{bbding}
\usepackage{pifont}
\usepackage{framed}
\usepackage{soul}
\usepackage{subfigure}
\usepackage{fourier-orns}
\usepackage{graphicx,booktabs,multirow,cuted,mathtools}
\usepackage{bm}
\usepackage{amsthm}
	\theoremstyle{definition}
		\newtheorem{lemma}{Lemma}
		\newtheorem{assumption}{Assumption}
		\newtheorem{theorem}{Theorem}

		\newtheorem{definition}{Definition}
		\newtheorem{remark}{Remark}
\usepackage{calligra}
\usepackage{pgfplots}
	\definecolor{colorhkust}{HTML}{142B8C}
	\definecolor{colorshanghaitech}{HTML}{A20005}
	\definecolor{colortsinghua}{HTML}{743481}
	\definecolor{colordark}{RGB}{184,134,11}
	\definecolor{colorRed}{RGB}{128, 0, 0}
	\definecolor{colorGreen}{RGB}{0, 64, 0}
	\definecolor{colorBlue}{RGB}{0, 0, 128}
\usepackage[
	colorlinks, 
	linkcolor=colorRed, 
	anchorcolor=blue, 
	citecolor=colorBlue, 
	urlcolor = colorGreen
]{hyperref}
\usepackage{tcolorbox}
\usepackage{changepage}

\begin{document}

\title{FedLoDrop: Federated LoRA with Dropout for Generalized LLM Fine-tuning
% \thanks{Identify applicable funding agency here. If none, delete this.}
}

\author
{Sijing Xie, Dingzhu Wen, Changsheng You, Qimei Chen, Mehdi Bennis, and Kaibin Huang
    \thanks{S. Xie and D. Wen are with the School of Information Science and Technology, ShanghaiTech University, Shanghai 201210, China (e-mail: \{xiesj2023, wendzh\}@shanghaitech.edu.cn). %\textit{(Corresponding author: Dingzhu Wen).}
    }
    \thanks{C. You is with the Department of Electronic and Electrical Engineering, Southern University of Science and Technology, Shenzhen 518055, China (e-mail: youcs@sustech.edu.cn).}
    \thanks{Q. Chen is with the School of Electronic Information, Wuhan University, Wuhan 430072, China (e-mail: chenqimei@whu.edu.cn).}
    \thanks{M. Bennis is with the Centre for Wireless Communications, University of Oulu, Oulu 90014, Finland (e-mail: mehdi.bennis@oulu.fi).}
    \thanks{K. Huang is with the Department of Electrical and Electronic Engineering, The University of Hong Kong, Hong Kong SAR, China (e-mail: huangkb@eee.hku.hk).}

}

\maketitle

\begin{abstract}
Fine-tuning (FT) large language models (LLMs) is crucial for adapting general-purpose models to specific tasks, enhancing accuracy and relevance with minimal resources. To further enhance generalization ability while reducing training costs, this paper proposes Federated LoRA with Dropout (FedLoDrop), a new framework that applies dropout to the rows and columns of the trainable matrix in Federated LoRA. A generalization error bound \textcolor{black}{and convergence analysis} under sparsity regularization are obtained, which elucidate the fundamental trade-off between underfitting and overfitting. The error bound reveals that a higher dropout rate increases model sparsity, thereby lowering the upper bound of pointwise hypothesis stability (PHS). While this reduces the gap between empirical and generalization errors, it also incurs a higher empirical error, which, together with the gap, determines the overall generalization error. On the other hand, though dropout reduces communication costs, deploying FedLoDrop at the network edge still faces challenges due to limited network resources. To address this issue, an optimization problem is formulated to minimize the upper bound of the generalization error, by jointly optimizing the dropout rate and resource allocation subject to the latency and per-device energy consumption constraints. To solve this problem, a branch-and-bound (B\&B)-based method is proposed to obtain its globally optimal solution. Moreover, to reduce the high computational complexity of the B\&B-based method, a penalized successive convex approximation (P-SCA)-based algorithm is proposed to efficiently obtain its high-quality suboptimal solution. Finally, numerical results demonstrate the effectiveness of the proposed approach in mitigating overfitting and improving the generalization capability.
\end{abstract}

\begin{IEEEkeywords}
     Large Language Models, Federated Learning, Low-rank Adaptation, Dropout, Generalization Error
\end{IEEEkeywords}

\section{Introduction}
Recent developments in large language models (LLMs), such as ChatGPT, LLaMA, and Vision Transformers, have driven significant progress in artificial general intelligence (AGI) \cite{shen2024large}. These models exhibit enhanced capabilities in generalization and inference. To adapt these models for specific tasks and data distributions, fine-tuning (FT) has become a key approach, which is crucial for establishing native artificial intelligence (AI) in the era of 6G \cite{10558825, jiang2024semantic,jiang2025comprehensive}. However, as model sizes grow, full-model FT becomes unaffordable in practice due to prohibitively high computational and memory costs.

To tackle these challenges, parameter-efficient fine-tuning (PEFT) has emerged as a pivotal approach \cite{wu2025adaptive, zhu2023prompt, bai2023knowprefix}. PEFT updates only a small subset of model parameters, providing a balanced trade-off between computational efficiency and training performance. One notable method within PEFT is Low-Rank Adaptation (LoRA) \cite{hu2021lora}, which employs low-rank matrices to approximate weight changes while keeping the original model weights frozen. The flexibility of LoRA allows it to be merged with the backbone, eliminating additional costs. Compared to other PEFT methods based on pruning, LoRA achieves superior efficiency-utility trade-offs \cite{kuo2024federated}.
% As the size of these models continues to grow, traditional fine-tuning methods, which require updating all model parameters, become increasingly computationally prohibitive and memory-intensive. In contrast, PEFT updates only a small subset of parameters or additional modules while keeping the rest of the model frozen, which can reduce the computation workload on devices. Key methodologies within the PEFT paradigm include the integration of adapters, which introduce lightweight, trainable modules within existing model architectures, allowing only these parameters to be fine-tuned while keeping the bulk of the model's weights frozen. Another prominent technique is Prompt Tuning, which focuses on optimizing the input prompts as a means to elicit desired outputs from the model, thus circumventing the need for extensive parameter adjustments. 
% decomposes weight updates into low-rank matrices, enabling efficient parameter modifications without the overhead of full model updates.
The authors in \cite{10472574} applied LoRA to a hierarchical LLM, classifying the instruction type and then utilizing task-specific networks to accomplish respective tasks. FT LLMs on a single device often result in suboptimal performance due to limited data and memory constraints \cite{10908556}. Centralized FT faces significant challenges, mainly because of privacy concerns and regulatory restrictions on data access \cite{10835069}. This is exacerbated by the fact that massive amounts of fragmented data are distributed across numerous devices at the network edge.

Federated learning (FL) enables privacy-preserving FT of pre-trained LLMs on distributed clients by sharing model updates between the server and clients, ensuring that distributed data remains localized \cite{11143883, 10453339, 11017442, 10909031, ZW_spectrum}. This makes FL an appealing choice for aligning LLMs with specialized domains \cite{10877919, 11026882}. Specifically, the authors in \cite{zhang2024towards, ye2024openfedllm} incorporated LoRA with FedAvg, significantly reducing the number of parameters that need to be synchronized across distributed devices. Several works have focused on improving the efficiency of federated LoRA. Due to the frequent exchange of LLM parameters among distributed devices, \cite{kuo2024federated} proposed integrating communication compression with federated LoRA to further reduce the communication cost. In addition, a hierarchical FedLoRA framework was proposed in \cite{liu2025resource}, which dynamically assigned diverse and suitable FT depths for each group, hence greatly reducing the computation and communication cost. The authors in \cite{sun2024improving} proposed a memory-efficient FT method, which sets the $\boldsymbol{A}$ matrices to fixed after initialization and trains only the $\boldsymbol{B}$ matrices. Other works focus on the deployment and performance of federated LoRA. \cite{babakniya2023slora} proposed an improved initialization strategy for LoRA’s weights, resulting in enhanced performance under the FedFM framework. Due to limited storage and computational capabilities at the edge devices, deploying the full model at edge devices is impractical. \cite{10855336} presented a split federated LoRA framework, where the computationally intensive encoder is deployed at the edge server, while 
others remain on edge devices.
% embedding and task modules remain on edge devices.

However, one significant challenge when FT models on downstream tasks is overfitting. Simply reducing the rank of LoRA could help alleviate overfitting, but fewer learnable parameters indicate less expressive power, and might lead to suboptimal performances. AdaLoRA \cite{zhangadaptive} optimizes LoRA by automatically pruning unimportant parameters with learned importance scores during training to prevent overfitting. However, this parameter selection method heavily relies on gradients of parameters on the training data, making the models less generalizable to unseen test data. To address this issue, a federated dropout framework was proposed in \cite{9707474}, which utilized the typical technique of dropout in deep learning. \cite{xie2024federated} further provided the convergence analysis for federated dropout, quantitatively showing the influence of dropout rate on convergence. In addition, the authors in \cite{lin2024lora} combined LoRA with dropout in an aggregation way. This, however, has the centralized challenges as mentioned before.

Towards this end, we propose a practical Federated LoRA with dropout (FedLoDrop) approach, which introduces random noise to the learnable low-rank matrices. FedLoDrop effectively mitigates overfitting while simultaneously reducing communication costs and minimizing update sizes transmitted from clients to the server \footnote{\textcolor{black}{Mixture-of-Experts (MoE) has recently emerged as an effective paradigm for reducing computational overhead and can serve as a complementary approach to FedLoDrop. The adaptive dropout mechanisms in FedLoDrop could be specifically tailored to promote more balanced expert utilization across heterogeneous clients, thereby preventing any single expert from dominating the learning process or remaining underutilized. }}. We provide the theoretical analysis using pointwise hypothesis stability (PHS) and Taylor expansion. 
% It is shown that a lower dropout rate enlarges the gap between generalization and empirical errors. 
Then, we aim to minimize the generalization error in each round by jointly optimizing dropout rate and network resource allocation. Finally, simulations are presented to validate the effectiveness of our proposed approach. The main contributions of this paper are listed below.
\begin{itemize}
    \item {\bf FedLoDrop Framework with Reduced Communication Cost}: We propose a practical FedLoDrop framework, where dropout is applied to the rows and columns of the tunable low-rank parameter matrices, e.g., $\boldsymbol{A}$ and $\boldsymbol{B}$. Specifically, for device $k$ with a dropout rate of $\gamma_k \in [0,1)$, dropout technique deactivates neurons on the input and output sides of both trainable LoRA matrices \cite{lin2024lora}. After applying dropout, the communication overhead is scaled down to $(1-\gamma)$ times of the original one.  

    \item {\bf Theoretical Analysis}: We characterize the effect of LoRA dropout on the sparsity and PHS upper bound. It is shown that a lower dropout rate enhances model complexity, increasing overfitting potential and widening the gap between generalization and empirical errors (adaptation function class). Conversely, a drastically high dropout rate may incur underfitting, significantly impairing the representation ability of the model and resulting in higher empirical error. This provides a theoretical foundation to balance the tradeoff between the adaptation function class and empirical error. \textcolor{black}{Moreover, a convergence analysis is also conducted, revealing that the convergence rate becomes slower with increasing dropout rate. }

    \item {\bf Joint Dropout Control and Resource Allocation}: Based on the theoretical analysis for FedLoDrop, we formulate an optimization problem to minimize the generalization error of each FT round under the network resource constraint, which is shown to be dependent on the dropout rate of each device. Particularly, a larger dropout rate leads to a smaller gap between empirical and generalization errors. As the exact form of the learning loss reduction in the generalization error bound is intractable, its upper bound is minimized instead without loss of generality, under the constraints of limited system subcarriers, completion latency, and per-device energy consumption. A branch-and-bound (B\&B)-based method is proposed to find the globally optimal solution. Moreover, to efficiently solve the problem, a low-complexity penalized successive convex approximation (P-SCA)-based solution is proposed to find a high-quality suboptimal solution.

    \item {\bf Performance Evaluation}: Extensive simulations based on multi-language tasks are conducted to evaluate the performance of the proposed schemes. We mainly FT two LLMs, RoBERTa-large (355M) and LLaMA (7B), on the GLUE and MMLU benchmarks, respectively. With dropout, the generalization ability of the fine-tuned model is significantly improved, thus enabling the capability to effectively apply the knowledge from the FT dataset to natural language response tasks. In both scenarios, more network resources, i.e., a longer latency, allow lower dropout rates of all devices, leading to improved testing performance.
\end{itemize}

% The remainder of this paper is organized as follows. 

\section{Framework of Federated LoRA Dropout}

\subsection{Preliminary-LoRA}

LoRA fine-tunes LLMs efficiently by maintaining the original model weights $\boldsymbol{\theta}_0$ frozen and adding low-rank trainable matrices  \cite{hu2021lora}. Specifically, the loss function of an LLM in the FT stage is
% \begin{equation}
%     \mathcal{L}=\frac{1}{|\mathcal{D}|}\sum_{i \in \mathcal{D} }\ell \left( \Delta \boldsymbol{\theta};\boldsymbol{\theta}_0, \boldsymbol{x}_i   \right),
% \end{equation}
$    \mathcal{L}=\frac{1}{|\mathcal{D}|}\sum_{i \in \mathcal{D} }\ell \left( \Delta \boldsymbol{\theta};\boldsymbol{\theta}_0, \boldsymbol{x}_i   \right),$
where $\mathcal{D}$ is the dataset, $|\mathcal{D}|$ is the size of dataset, $\boldsymbol{x}_i$ is the $i$-th data sample therein, and $\ell(\cdot;\cdot)$ is the empirical loss function that characterizes the difference between the output and real label.

Consider an arbitrary layer in an arbitrary training round $\boldsymbol{W}_{u,t}$ applied with LoRA for FT,  a LoRA update is thus characterized by a set of low-rank trainable weights $\Delta \boldsymbol{\theta}_t \triangleq \{ \Delta \boldsymbol{W}_{u,t} \}_{u=1}^{U^{\prime}} $, a set of pre-trained weight $\boldsymbol{\theta}_0 \triangleq \{ \boldsymbol{W}_{u,0} \}_{u=1}^U$, where $U^{\prime}$ is the number of weight matrices applying LoRA and $U$ is the number of all matrices. LoRA may not update all matrices, in which case $U^{\prime} \leq U$. The parameter matrix $\boldsymbol{W}_{u,t}\in \mathbb{R}^{n_1 \times n_2}$ can be regarded as the sum of a frozen pre-trained matrix $\boldsymbol{W}_{u,0} \in \mathbb{R}^{n_1 \times n_2}$ and a trainable low-rank decomposable matrix $\Delta \boldsymbol{W}_{u,t} \in \mathbb{R}^{n_1 \times n_2}$, i.e.,
\begin{equation}
    \begin{aligned}
            \boldsymbol{W}_{u,t} &= 
            \boldsymbol{W}_{u,0}+\Delta \boldsymbol{W}_{u,t}=\boldsymbol{W}_{u,0}+\boldsymbol{B}_{u,t}\boldsymbol{A}_{u,t},
    \end{aligned}
\label{equ:lora}
\end{equation}
% {\color{red}$           \boldsymbol{W}_u = 
%             \boldsymbol{W}_{u,0}+\Delta \boldsymbol{W}_u=\boldsymbol{W}_{u,0}+\boldsymbol{B}_u\boldsymbol{A}_u,$}
where  $\boldsymbol{A}_{u,t} \in \mathbb{R}^{r \times n_2}$ and  $\boldsymbol{B}_{u,t} \in \mathbb{R}^{n_1 \times r}$ are trainable low-rank matrices, with $r \ll \{n_1,n_2\}$. For an arbitrary training round $t$ and data sample $i$, the training procedure of LoRA is presented below.

\begin{figure*}[htbp]
    \centering
    \includegraphics[width=0.7\linewidth]{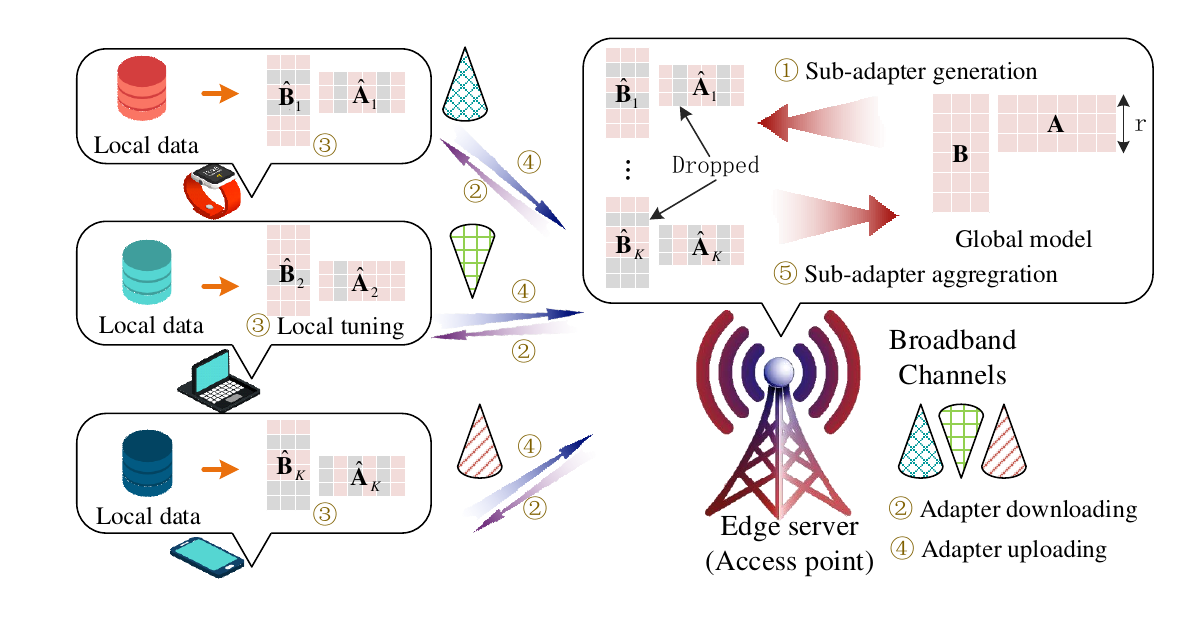}
    \caption{The operations of FedLoDrop in a wireless system.}
    \label{fig:system}\vspace{-0.5cm}
\end{figure*}

\subsubsection{Forward Pass} 
The output of the $u$-th layer is given by 
\begin{equation}
        \begin{aligned}
            {\boldsymbol{h}}_{u,i,t} &= \boldsymbol{W}_{u,t-1}{\boldsymbol{f}_{u-1,i,t}}=(\boldsymbol{W}_{u,0}+\Delta \boldsymbol{W}_{u,t-1}) \boldsymbol{f}_{u-1,i,t}\\
            & = (\boldsymbol{W}_{u,0}+{\boldsymbol{B}_{u,t-1}}{\boldsymbol{A}_{u,t-1}}) \boldsymbol{f}_{u-1,i,t},
    \end{aligned}
\label{equ:lora_for}
\end{equation}
where $\boldsymbol{f}_{u-1,i,t}$ is the output of the previous layer.
% based on $\boldsymbol{x}_i$ in $t$-th round.

\subsubsection{Backward Pass}
During the backward propagation, the stochastic gradients of the two low-rank matrices $\textbf{B}$ and $\textbf{A}$ are calculated individually, given by
\begin{equation}\label{equ:backpropa_pri_1}
\begin{aligned}
     \Delta\boldsymbol{B}_{u,i,t-1} 
     &=\frac{\partial \ell_{t}}{\partial {\boldsymbol{B}_{u,t-1}}} = \frac{\partial \ell_{t}}{\partial  {\boldsymbol{h}_{u,i,t}}}  \frac{\partial  {\boldsymbol{h}_{u,i,t}}}{\partial {\boldsymbol{B}_{u,t-1}}} 
     \\
     & = \frac{\partial \ell_{t}}{\partial  {\boldsymbol{h}_{u,i,t}}} \cdot \left( {\boldsymbol{A}_{u,t-1}} \boldsymbol{f}_{u-1,i,t}\right)^{\top},
\end{aligned}
\end{equation}

\begin{equation}\label{equ:backpropa_pri_2}
\begin{aligned}
    \Delta\boldsymbol{A}_{u,i,t-1} 
     &= \frac{\partial \ell_{t}}{\partial {\boldsymbol{A}_{u,t-1}}} = \frac{\partial \ell_{t}}{\partial  {\boldsymbol{h}_{u,i,t}}}  \frac{\partial  {\boldsymbol{h}_{u, i,t}}}{\partial {\boldsymbol{A}_{u,t-1}}} \\
     & =  {\boldsymbol{B}^{\top}_{u,t-1}}  \cdot \frac{\partial \ell_{t}}{\partial  {\boldsymbol{h}_{u,i,t}}} \cdot   \boldsymbol{f}^{\top}_{u-1,i,t}.
\end{aligned}
\end{equation}
% for all ${\bf x}_i \in \mathcal{D}$.

\subsubsection{Low-rank Matrices Updating}

$\boldsymbol{B}_{u,t}$ and $\boldsymbol{A}_{u,t}$ are updated by the aggregated stochastic gradients of all data samples \cite{lin2024lora,10807365,cho2023heterogeneous}, i.e.,
%, we calculate the gradients of $\boldsymbol{B}_t$ and $\boldsymbol{A}_t$ by the aggregated loss,
\begin{equation}\label{equ:backpropa_1}
\begin{aligned}
     \Delta\boldsymbol{B}_{u,t-1} 
     &=\frac{\partial \mathcal{L}}{\partial {\boldsymbol{B}_{u,t-1}}} = \frac{1}{|\mathcal{D}|}\sum_{i \in \mathcal{D} } \frac{\partial \ell_{t}}{\partial {\boldsymbol{B}_{u,t-1}}} \\
     & = \frac{1}{|\mathcal{D}|}\sum_{i \in \mathcal{D} } \Delta\boldsymbol{B}_{u,i,t-1},
\end{aligned}
\end{equation}
\begin{equation}\label{equ:backpropa_2}
\begin{aligned}
    \Delta\boldsymbol{A}_{u,t-1} 
     &= \frac{\partial \mathcal{L}}{\partial {\boldsymbol{A}_{u,t-1}}} = \frac{1}{|\mathcal{D}|}\sum_{i \in \mathcal{D} } \frac{\partial \ell_{t}}{\partial {\boldsymbol{A}_{u,t-1}}}\\
     & = \frac{1}{|\mathcal{D}|}\sum_{i \in \mathcal{D} } \Delta\boldsymbol{A}_{u,i,t-1},
\end{aligned}
\end{equation}
where $\Delta\boldsymbol{B}_{u,i,t-1}$ and $\Delta\boldsymbol{A}_{u,i,t-1}$ are the sample-wise stochastic gradient matrices of the $u$-th layer defined in  \eqref{equ:backpropa_pri_1}, \eqref{equ:backpropa_pri_2}. 
%for each data sample, $\Delta\boldsymbol{B}_{i,t}$ and $\Delta\boldsymbol{A}_{i,t}$ are calculated as Eq. \eqref{equ:backpropa_pri}, 
Then ${\boldsymbol{B}_{u,t}}$ and ${\boldsymbol{A}_{u,t}}$ are updated through gradient descent, with $\alpha_B$ and $\alpha_A$ being the learning rate:
\begin{equation}\label{equ:sgd_1}
    \begin{aligned}
            \boldsymbol{B}_{u,t} &= \boldsymbol{B}_{u,t-1} - \alpha_B\Delta\boldsymbol{B}_{u,t-1},
    \end{aligned}
\end{equation}
\begin{equation}\label{equ:sgd_2}
    \begin{aligned}
            \boldsymbol{A}_{u,t} &= \boldsymbol{A}_{u,t-1} - \alpha_A\Delta\boldsymbol{A}_{u,t-1}.
    \end{aligned}
\end{equation}
Next, the trainable matrix $\Delta \boldsymbol{W}_{u,t}$ is updated by
\begin{equation}\label{Eq:UpdateMatrix}
\Delta \boldsymbol{W}_{u,t}=\boldsymbol{B}_{u,t}\boldsymbol{A}_{u,t},
\end{equation}
and the parameter matrix $\boldsymbol{W}_{u,t}$ is updated by
\begin{equation}\label{Eq:UpdateW}
\boldsymbol{W}_{u,t}=\boldsymbol{W}_{u,0}+\Delta \boldsymbol{W}_{u,t}.
\end{equation}

\begin{remark}
The low-rank matrices updating method presented in \eqref{equ:backpropa_1} -- \eqref{Eq:UpdateMatrix} is not equivalent to the sample average of the updates of $\Delta {\bf W}_u$, i.e.,
\begin{equation*}
        \begin{aligned}
        & \Delta \boldsymbol{W}_{u,t} = \left( \boldsymbol{B}_{u,t-1} - \frac{\alpha_B }{|\mathcal{D}|}\sum_{i \in \mathcal{D} } \Delta \boldsymbol{B}_{u,i,t-1} \right) \\
                      & \left( \boldsymbol{A}_{u,t-1} - \frac{\alpha_A }{|\mathcal{D}|}\sum_{i \in \mathcal{D} } \boldsymbol{A}_{u,i,t-1} \right)\\
                     & \neq \frac{1}{|\mathcal{D}|}\sum_{i \in \mathcal{D} } (\boldsymbol{B}_{u,t-1} - \alpha_{B,i}\Delta\boldsymbol{B}_{u,i,t-1})\\
                     & (\boldsymbol{A}_{u,t-1} - \alpha_{A,i}\Delta\boldsymbol{A}_{u,i,t-1} ).
        \end{aligned}
    \end{equation*}
However, this enjoys two benefits. On one hand, it simplifies the initialization of $\boldsymbol{B}$ and $\boldsymbol{A}$ by utilizing the ones in the last training round. On the other hand, this approach is well-suited for distributed implementations \cite{zhang2024towards, babakniya2023slora}. In each round, devices only need to transmit $\boldsymbol{B}_{u,t-1}$ and $\boldsymbol{A}_{u,t-1}$, which have a lower dimensionality compared to $\Delta\boldsymbol{W}_{u,t-1}$, thereby reducing the communication overhead.
\end{remark}

\subsection{LoRA Dropout}\label{Sect:LoraDropout}
The original dropout technique was proposed in \cite{srivastava2014dropout} to avoid overfitting during training. In standard dropout, each neuron in the network is dropped from the network with a certain probability. Given that dropout techniques have proven effective in controlling overfitting, we introduce a LoRA dropout framework to enhance generalization when adapting to downstream tasks. The loss function can be re-written as
% \begin{equation}
%     \mathcal{L}=\frac{1}{|\mathcal{D}|}\sum_{i \in \mathcal{D} }\ell \left(\Delta \boldsymbol{\theta}(\boldsymbol{m}_t); \boldsymbol{\theta}_0 , {x}_i \right),
% \end{equation}
$    \mathcal{L}=\frac{1}{|\mathcal{D}|}\sum_{i \in \mathcal{D} }\ell \left(\Delta \boldsymbol{\theta}(\boldsymbol{m}_t); \boldsymbol{\theta}_0 , {x}_i \right),$
where $\Delta {\boldsymbol{\theta}}(\boldsymbol{m}_t)$ is the LoRA parameters after the dropout,  $\boldsymbol{m}_t$ is the concatenation of masks of all LoRA modules, which varies across rounds, and $\boldsymbol{\theta}_0$ is the original parameters of the pre-trained model. 

As shown step \normalsize{\textcircled{\scriptsize{1}}}\normalsize in Fig. \ref{fig:system}, for a LoRA module, rows and columns are randomly dropped from both tunable low-rank parameter matrices with a probability of $\gamma_t$, which is called dropout rate. In other words, LoRA dropout strategy samples random neurons on the input and output sides of LoRA matrices with a probability $\gamma_t$ to mask them to zeros \cite{lin2024lora}. The low-rank matrices after dropout are given by
\begin{equation}
\left\{
\begin{aligned}
    &\hat{\boldsymbol{A}}_{u,t} = \boldsymbol{A}_{u,t} \cdot \mathrm{diag}(\boldsymbol{m}_{A,t}), \; \boldsymbol{m}_{A,t} \sim \text{Bern}(1-\gamma_t), \\
    &\hat{\boldsymbol{B}}_{u,t} = \left(\boldsymbol{B}_{u,t}^\top \cdot \mathrm{diag}(\boldsymbol{m}_{B,t})\right)^\top, \; \boldsymbol{m}_{B,t} \sim \text{Bern}(1-\gamma_t),
\end{aligned}
\right.
\label{equ:a&b}
\end{equation}
where $\boldsymbol{m}_{A,t} \in \mathbb{R}^{n_2}$, $\boldsymbol{m}_{B,t} \in \mathbb{R}^{n_1}$, $\text{Bern}(1-\gamma_t)$ means that all elements of a matrix are independently distributed as the identical Bernoulli distribution with parameter $(1-\gamma_t)$.
With dropout, the sizes of $\hat{\boldsymbol{A}}_{u,t}$ and $\hat{\boldsymbol{B}}_{u,t}$ are further reduced. In the $t$-th training round, for the $i$-th data sample, the training procedure is presented below.
\subsubsection{Forward Pass}
After dropout, the output of the $u$-th layer is given by $\hat{\boldsymbol{h}}_{u,i,t}=(\boldsymbol{W}_{u,0}+\hat{\boldsymbol{B}}_{u,t-1}\hat{\boldsymbol{A}}_{u,t-1})\boldsymbol{f}_{u-1,i,t}.$
% \begin{equation}\label{equ:lorafor_drop}
%         \begin{aligned}
%             \hat{\boldsymbol{h}}_{u,i,t} &=(\boldsymbol{W}_{u,0}+\hat{\boldsymbol{B}}_{u,t-1}\hat{\boldsymbol{A}}_{u,t-1})\boldsymbol{f}_{u-1,i,t}.
%     \end{aligned}
% \end{equation}

\subsubsection{Backward Pass and Low-rank Matrices Updating}
During backpropagation, the stochastic gradients of $\hat{\boldsymbol{B}}_{u,t}$ and $\hat{\boldsymbol{A}}_{u,t}$ are computed according to \eqref{equ:backpropa_1},  \eqref{equ:backpropa_2} and updated following \eqref{equ:sgd_1}, \eqref{equ:sgd_2}. Notably, the dropped parameters will not be updated, thereby reducing the computational overhead. After obtaining $\hat{\boldsymbol{B}}_{u,t}$ and $\hat{\boldsymbol{A}}_{u,t}$, the matrices $\Delta\hat{\boldsymbol{W}}_{u,t}$ and $\hat{\boldsymbol{W}}_{u,t}$ are calculated in the same manner as \eqref{Eq:UpdateMatrix} and \eqref{Eq:UpdateW}.

% \begin{equation}\label{equ:backpropa}
% \begin{aligned}
%     & \frac{\partial \mathcal{L}}{\partial \hat{\boldsymbol{B}}} = \frac{\partial \mathcal{L}}{\partial  \hat{\boldsymbol{h}}} \cdot \frac{\partial  \hat{\boldsymbol{h}}}{\partial \hat{\boldsymbol{B}}} = \frac{\partial \mathcal{L}}{\partial  \hat{\boldsymbol{h}}} \cdot ( \hat{\boldsymbol{A}} \boldsymbol{x})^{\top} = \frac{\partial \mathcal{L}}{\partial  {\boldsymbol{W}}} \cdot  \hat{\boldsymbol{A}} ^{\top}, \\
%     & \frac{\partial \mathcal{L}}{\partial \hat{\boldsymbol{A}}} = \frac{\partial \mathcal{L}}{\partial  \hat{\boldsymbol{h}}} \cdot \frac{\partial  \hat{\boldsymbol{h}}}{\partial \hat{\boldsymbol{A}}} = \hat{\boldsymbol{B}} ^{\top} \cdot \frac{\partial \mathcal{L}}{\partial  \hat{\boldsymbol{h}}} \cdot   \boldsymbol{x}^{\top} =  \hat{\boldsymbol{B}} ^{\top} \cdot \frac{\partial \mathcal{L}}{\partial  {\boldsymbol{W}}}.
% \end{aligned}
% \end{equation}
% Finally, SGD is utilized to update $\hat{\boldsymbol{B}}$ and $\hat{\boldsymbol{A}}$.

\subsection{Federated LoRA with Dropout}
\subsubsection{Procedure and Algorithm}
In this case, each device $k$ holds a local dataset, denoted by $\mathcal{D}_k = \{\boldsymbol{x}_i | i=1,2,\ldots,|\mathcal{D}_k| \}$, where $\boldsymbol{x}_i$ is the $i$-th data sample and the size of $\mathcal{D}_k$ is $|\mathcal{D}_k|$. Denote $\Delta {\boldsymbol{\theta}}(\boldsymbol{m}_{k,t})$ as the LoRA parameters after the dropout for the $k$-th device, $\boldsymbol{m}_{k,t}$ as the concatenation of all dropout masks of LoRA modules, which
varies across different training rounds and different devices, and $\boldsymbol{\theta}_0$ as the original parameters of the pre-trained model, which has been stored in each device in advance. The objective is to minimize the global loss function, given by
% \begin{equation}
%     \mathcal{L}=\sum_{k=1}^K \frac{\left|\mathcal{D}_k\right|}{|\mathcal{D}|} \ell_{k}\left( \Delta {\boldsymbol{\theta}}(\boldsymbol{m}_{k,t}); \boldsymbol{\theta}_0 ,\mathcal{D}_k  \right),
% \end{equation}
$    \mathcal{L}=\sum_{k=1}^K \tfrac{\left|\mathcal{D}_k\right|}{|\mathcal{D}|} \ell_{k}\left( \Delta {\boldsymbol{\theta}}(\boldsymbol{m}_{k,t}); \boldsymbol{\theta}_0 ,\mathcal{D}_k  \right),$
where $\mathcal{D} = \{\mathcal{D}_k\}$ is the global dataset, the local loss function of the $k$-th device can be written as
% \begin{equation}
%     \begin{aligned}
%            & \ell_{k}\left( \Delta {\boldsymbol{\theta}}(\boldsymbol{m}_{k,t}); \boldsymbol{\theta}_0, \mathcal{D}_k  \right) \\
%            &=\frac{1}{|\mathcal{D}_k|}\sum_{{x}_i\in\mathcal{D}_k}\ell_{k,i} \left(\Delta {\boldsymbol{\theta}}(\boldsymbol{m}_{k,t}); \boldsymbol{\theta}_0, \boldsymbol{x}_i \right).
%     \end{aligned}
% \end{equation} 
$            \ell_{k}\left( \Delta {\boldsymbol{\theta}}(\boldsymbol{m}_{k,t}); \boldsymbol{\theta}_0, \mathcal{D}_k  \right) =\tfrac{1}{|\mathcal{D}_k|}\sum_{{x}_i\in\mathcal{D}_k}\ell_{k,i} \left(\Delta {\boldsymbol{\theta}}(\boldsymbol{m}_{k,t}); \boldsymbol{\theta}_0, \boldsymbol{x}_i \right).$
Similar to \eqref{equ:a&b}, FedLoDrop can be written as
\begin{equation}
\begin{aligned}
    &\hat{\boldsymbol{A}}_{u,k,t} = \boldsymbol{A}_{u,t} \cdot \mathrm{diag}(\boldsymbol{m}_{k,A,t}), \boldsymbol{m}_{k,A,t} \sim \text{Bern}(1-\gamma_{k,t}), \\
    &\hat{\boldsymbol{B}}_{u,k,t} = \left(\boldsymbol{B}_{u,t}^\top \cdot \mathrm{diag}(\boldsymbol{m}_{k,B,t})\right)^\top, \boldsymbol{m}_{k,B,t} \sim \text{Bern}(1-\gamma_{k,t}),
\end{aligned}
\label{equ:ab_Fed:LoDrop} 
\end{equation}

The forward and backward pass on each device has minor changes to the aforementioned situations and thus are omitted. 

As illustrated in Fig. \ref{fig:system} and Algorithm \ref{Alg:wholeprocess}, the overall framework involves two primary components: local tuning operations on the client side, and dropout and aggregation operations on the server side, which work together to ensure efficient training. There are five steps to complete an arbitrary training round $t$, as described below.
\begin{itemize}
\item \emph{LoRA Sub-adapter Generation}: The server adopts the LoRA dropout technique introduced in Section \ref{Sect:LoraDropout} to generate a sub-adapter for each device. \textcolor{black}{Crucially, the dropout masks are statistically independent across rounds and generated prior to any client's local computation, thus preventing privacy  leaks.}

\item \emph{Adapter Downloading}: \textcolor{black}{Each device downloads its corresponding sub-adapter ${\hat{\boldsymbol{B}}}_{k,t}$ and ${\hat{\boldsymbol{A}}}_{k,t}$ from the server, where ${\hat{\boldsymbol{B}}}_{k,t} \triangleq 
 \{ {\hat{\boldsymbol{B}}}_{u,k,t} \}_{u=1}^{U^{\prime}} $, ${\hat{\boldsymbol{A}}}_{k,t} \triangleq 
 \{ {\hat{\boldsymbol{A}}}_{u,k,t} \}_{u=1}^{U^{\prime}} $.}

\item \emph{Client Local Tuning}: Each device first calculates the product of two low-rank adapters and updates them based on pre-trained model and its own datasets, as shown in  \eqref{equ:backpropa_1} -- \eqref{equ:sgd_2}. 

% The local gradient vector of device $k$ is denoted as $\tilde{\bf g}_k^{(t)}$.

\item \emph{Adapter Uploading}: \textcolor{black}{Each device uploads the updated gradients of local sub-adapters to the server by the allocated subcarriers. }

\item \emph{Server Aggregation}: \textcolor{black}{The server reconstructs the full-size adapter update via zero-padding: it inserts zeros into the dropped positions. This process is feasible because the server has knowledge of the sparsity mask. Then, all complete networks are aggregated for updating the global network, expressed as:}
% \begin{equation}\label{equ:aggre}
%     \begin{aligned}
%          \hat{\boldsymbol{B}}_t = \sum_{k=1}^K \frac{\left|\mathcal{D}_k\right|}{|\mathcal{D}|} \hat{\boldsymbol{B}}_{k,t}, \quad \hat{\boldsymbol{A}}_t = \sum_{k=1}^K \frac{\left|\mathcal{D}_k\right|}{|\mathcal{D}|} \hat{\boldsymbol{A}}_{k,t}.
%     \end{aligned}
% \end{equation}
\begin{equation}\label{equ:aggre_1}
    \begin{aligned}
          {\boldsymbol{B}}_t &= {\boldsymbol{B}}_{t-1} - \alpha_B\sum_{k=1}^K \frac{\left|\mathcal{D}_k\right|}{|\mathcal{D}|} \Delta\hat{\boldsymbol{B}}_{k,t-1}\\
         & = {\boldsymbol{B}}_{t-1} - \alpha_B \Delta\hat{\boldsymbol{B}}_{t-1}, 
    \end{aligned}
\end{equation}
\begin{equation}\label{equ:aggre_2}
    \begin{aligned}
          {\boldsymbol{A}}_t & = {\boldsymbol{A}}_{t-1} -\alpha_A\sum_{k=1}^K \frac{\left|\mathcal{D}_k\right|}{|\mathcal{D}|} \Delta\hat{\boldsymbol{A}}_{k,t-1}\\
         & = {\boldsymbol{A}}_{t-1} - \alpha_A \Delta\hat{\boldsymbol{A}}_{t-1}.
    \end{aligned}
\end{equation}
% $\hat{\boldsymbol{B}}_t = \sum_{k=1}^K \frac{\left|\mathcal{D}_k\right|}{|\mathcal{D}|} \hat{\boldsymbol{B}}_{k,t}$ and $\hat{\boldsymbol{A}}_t = \sum_{k=1}^K \frac{\left|\mathcal{D}_k\right|}{|\mathcal{D}|} \hat{\boldsymbol{A}}_{k,t}$. 
\end{itemize}
% The above steps iterate until model convergence or reaches the preset number of global iterations.

% As shown before, each device only needs to send two dropped and low-rank matrices in each round, further decreasing the communication overhead. 
\begin{remark}
    Combined with LoRA dropout, the inference caused by the cross-product of sub-adapters from different clients can be diminished. \textcolor{black}{Furthermore, FedLoDrop inherently provides a degree of privacy due to the sparsification of transmitted updates against potential privacy attacks such as gradient inversion or model extraction. That said, the dropout-based sparsification in FedLoDrop may offer incidental protection against gradient inversion attacks, as partially masked updates reduce the amount of information available to an adversary. Similarly, the local gradient tracking mechanism further decouples client-specific information from the global model, potentially mitigating model extraction risks.}
\end{remark}
%\vspace{-0.5cm}

\begin{algorithm}[t]
	\caption{The training process of FedLoDrop}\label{Alg:wholeprocess}
	\LinesNumbered
\begin{adjustwidth}{}{+0.4cm}
\begin{tcolorbox}[
    colback=gray!10,      % 背景颜色
    colframe=white,       % 边框颜色
    boxsep=0pt,           % 内容与边框之间的距离
    left=0pt,             % 左边内边距
    right=1pt,            % 右边内边距
    top=0pt,              % 上边内边距
    bottom=0pt,           % 下边内边距
    arc=0mm,              % 圆角弧度
    boxrule=0.1pt         % 边框宽度
]
        \textbf{Parameters:} Communication round $\emph{T}$; The pre-trained model $\boldsymbol{W}_0$; The local trainable and efficient parameters $\boldsymbol{B}_k$, $\boldsymbol{A}_k$ and the local dataset $\mathcal{D}_k$ of the $k$-th device.\\
\end{tcolorbox}
\end{adjustwidth}    
\begin{adjustwidth}{}{+0.4cm}
\begin{tcolorbox}[
    colback=gray!13,      % 背景颜色
    colframe=white,       % 边框颜色
    boxsep=0pt,           % 内容与边框之间的距离
    left=0pt,             % 左边内边距
    right=1pt,            % 右边内边距
    top=0pt,              % 上边内边距
    bottom=0pt,           % 下边内边距
    arc=0mm,              % 圆角弧度
    boxrule=0.1pt         % 边框宽度
]        
        \textbf{Before Training:} Store $\boldsymbol{W}_0$ on each device, and initialize $\boldsymbol{B}_0$, and $\boldsymbol{A}_0$ on the server.\\
\end{tcolorbox}
\end{adjustwidth}
\begin{adjustwidth}{}{+0.4cm}
\begin{tcolorbox}[
    colback=purple!13,      % 背景颜色
    colframe=white,       % 边框颜色
    boxsep=0pt,           % 内容与边框之间的距离
    left=0pt,             % 左边内边距
    right=1pt,            % 右边内边距
    top=0pt,              % 上边内边距
    bottom=0pt,           % 下边内边距
    arc=0mm,              % 圆角弧度
    boxrule=0.1pt         % 边框宽度
]        
        \textbf{Server executes:}\\
        \For{each communication round $t = 1$ to $\emph{T}$}{
        $\hat{\boldsymbol{B}}_{k,t-1}$, $\hat{\boldsymbol{A}}_{k,t-1}$ $\leftarrow$ (LoRA dropout)\\
        Send $\hat{\boldsymbol{B}}_{k,t-1}$, $\hat{\boldsymbol{A}}_{k,t-1}$ to each device\\
        \For{each device \textbf{in parallel}}{ClientLocalTuning ($k$, $\hat{\boldsymbol{B}}_{k,t-1}$, $\hat{\boldsymbol{A}}_{k,t-1}$)}
        Receive local updated parameters \\
        % $\hat{\boldsymbol{B}}_{k,t}$, $\hat{\boldsymbol{A}}_{k,t}$\\
        Global Aggregation by zero padding and \eqref{equ:aggre_1}, \eqref{equ:aggre_2}   
        }
\end{tcolorbox}   
\end{adjustwidth}
\begin{adjustwidth}{}{+0.4cm}
\begin{tcolorbox}[
    colback=blue!7,      % 背景颜色
    colframe=white,       % 边框颜色
    boxsep=0pt,           % 内容与边框之间的距离
    left=0pt,             % 左边内边距
    right=1pt,            % 右边内边距
    top=0pt,              % 上边内边距
    bottom=0pt,           % 下边内边距
    arc=0mm,              % 圆角弧度
    boxrule=0.1pt         % 边框宽度
] 
        \textbf{ClientLocalTuning} ($k$, $\hat{\boldsymbol{B}}_{k,t-1}$, $\hat{\boldsymbol{A}}_{k,t-1}$)\textbf{:}\\
        $\hat{\boldsymbol{B}}_{k,t}$, $\hat{\boldsymbol{A}}_{k,t}$ $\leftarrow$ \eqref{equ:backpropa_1} -- \eqref{equ:sgd_2} \\
        Send updated parameters
        % $\hat{\boldsymbol{B}}_{k,t}$, $\hat{\boldsymbol{A}}_{k,t}$ 
        to the server
\end{tcolorbox}
\end{adjustwidth}
\end{algorithm}

\subsubsection{Communication-computation-memory Overhead}
\begin{itemize}
    \item \emph{Communication overhead}: FedLoDrop provides significant competitive savings at both uploading and downloading links. In round $t$, each device downloads and uploads ${\hat{\boldsymbol{B}}}_{k,t}$ and ${\hat{\boldsymbol{A}}}_{k,t}$ to the server for aggregation. The original number of transmitted parameters is $M=(n_1 + n_2)r$. After applying dropout, with a dropout rate of $\gamma_k \in [0,1)$ for device $k$, the communication overhead becomes $\hat{M}_{k,t}=(1-\gamma_{k,t})(n_1 + n_2)r= (1-\gamma_{k,t})M.$
    % \begin{equation}
    %     \hat{M}_{k,t}=(1-\gamma_{k,t})(n_1 + n_2)r= (1-\gamma_{k,t})M_{\text{ori}}.
    % \end{equation}

    \item \emph{Computation and memory overhead}:
    FedLoDrop saves local computation workloads and memory cost by maintaining sparse adapters throughout the FL process. No part of FedLoDrop requires dense training, and the computation overhead is conducted sample-wise. As the computation and memory costs of adapters are small compared to the costs of the backbone and dropout only reduces computation overhead in backward propagation, the computation and memory benefits are omitted \cite{kuo2024federated}.
\end{itemize}

\section{Theoretical Analysis}
In the FT phase, we first formulate an optimization problem to minimize the loss function under the constraint of model sparsity. Subsequently, we characterize a generalization error bound within the framework of sparsity regularization, which demonstrates a fundamental trade-off between underfitting and overfitting in the context of LoRA dropout. Next, the upper bounding loss descent is formulated. \textcolor{black}{Finally, the convergence of FedLoDrop is performed. }

\subsection{LoRA Dropout FT Through Sparse Regularization}
Building upon the LoRA dropout mechanism presented on a single device by \cite{lin2024lora}, we extend this framework to accommodate FL environments. Assume $\boldsymbol{d}_{k,t} \in \{0,1 \}$ as a dropout instance applied to the production of LoRA matrices, i.e., $\boldsymbol{d}_{k,t} \sim \text{Bern}(1-(1-\gamma_{k,t})^2)$, namely, $\boldsymbol{d}_{k,t} \sim \text{Bern}(2\gamma_{k,t}-\gamma_{k,t}^2)$, where $\boldsymbol{d}_{k,t}$ equals to one when the corresponding entry is dropped to zero and set as zero otherwise. The FT to minimize the loss function can be formulated as:
\begin{equation}\label{problem_originial}
	\begin{array}{cl}
		\underset
		{ \{ \Delta {\boldsymbol{\theta}}_{k,t} \}	}
		{
			\min
		}
		&	
		 \sum_{k=1}^K \frac{\left|\mathcal{D}_{k}\right|}{|\mathcal{D}|} \ell_{k}\left(\Delta {\boldsymbol{\theta}}_{k,t}; \boldsymbol{\theta}_0, \mathcal{D}_k  \right),
		\\
		\text{subject to}
		& \mathbb{E}_{\boldsymbol{d}_{k,t} \sim \text{Bern}(2\gamma_{k,t}-\gamma_{k,t}^2)} \| d_{k,t} \odot \Delta {\boldsymbol{\theta}}_{k,t} \|_2^2 \leq c, \forall k,
	\end{array}
\end{equation}
where $c$ is a constant, and the condition denotes the sparsity of $\Delta \boldsymbol{\theta}_{k,t}$. 
% The original objective function is converted to 
% $$\min_{ \{ \Delta {\boldsymbol{\theta}}_{k,t} \}} \mathcal{L}^{\prime}= \sum_{k=1}^K \frac{\left|\mathcal{D}_k\right|}{|\mathcal{D}|} \min_{ \{\Delta {\boldsymbol{\theta}}_{k,t} \}} \ell_{k}\left(\Delta {\boldsymbol{\theta}}_{k,t}; \boldsymbol{\theta}_0, \mathcal{D}_k  \right). $$
% %\noindent[{\color{red}XXX: The usage of notation is bad. The derivation logic needs improvement. The physical meaning of (20) requires explanation.}]
% $c$ is eliminated as a constant, and 
The regularized optimization problem of the global one is

% \begin{equation}\label{equ:relation_loss}
%     \begin{aligned}
%             \mathcal{L}^{\prime}_{\{\nu_{k,t} \}} 
%             &= \sum_{k=1}^K \frac{\left|\mathcal{D}_k\right|}{|\mathcal{D}|} \left( \min_{ \{ \Delta {\boldsymbol{\theta}}_{k,t} \}}   \ell_{k}\left(\Delta {\boldsymbol{\theta}}_{k,t}; \boldsymbol{\theta}_0, \mathcal{D}_k \right) \right.\\
%             & \left. \quad + \sum_{k=1}^K\nu_{k,t} \mathbb{E}_{\boldsymbol{d}_{k,t} \sim \text{Bern}(2\gamma_{k,t}-\gamma_{k,t}^2)} \| d_{k,t} \odot \Delta {\boldsymbol{\theta}}_{k,t} \|_2^2 \right)\\
%             & = \sum_{k=1}^K \frac{\left|\mathcal{D}_k\right|}{|\mathcal{D}|}  \ell_{k,\lambda_t},\\
%     \end{aligned} 
% \end{equation}
\begin{equation}\label{equ:relation_loss}
    \begin{aligned}
            & \mathcal{L}^{\prime}_{\{\nu_{k,t} \}} 
            = \sum_{k=1}^K \frac{\left|\mathcal{D}_k\right|}{|\mathcal{D}|} \left( \min_{ \{ \Delta {\boldsymbol{\theta}}_{k,t} \}}   \ell_{k}\left(\Delta {\boldsymbol{\theta}}_{k,t}; \boldsymbol{\theta}_0, \mathcal{D}_k \right) \right.+\\
            & \left.  \sum_{k=1}^K\nu_{k,t} \mathbb{E}_{\boldsymbol{d}_{k,t} \sim \text{Bern}(2\gamma_{k,t}-\gamma_{k,t}^2)} \| d_{k,t} \odot \Delta {\boldsymbol{\theta}}_{k,t} \|_2^2 \right) = \sum_{k=1}^K \frac{\left|\mathcal{D}_k\right|}{|\mathcal{D}|}  \ell_{k,\lambda_t},\\
    \end{aligned} 
\end{equation}
where $\ell_{k,\lambda_t} = \min_{\Delta \boldsymbol{\theta}_{k,t}} \ell_{k}\left(\Delta {\boldsymbol{\theta}}_{k,t}; \boldsymbol{\theta}_0, \mathcal{D}_k  \right) + \lambda_t \mathbb{E}_{\boldsymbol{d}_{k,t} \sim \text{Bern}(2\gamma_{k,t}-\gamma_{k,t}^2)} \| d_{k,t} \odot \Delta \boldsymbol{\theta}_{k,t} \|_2^2 $ is the regularized optimization problem for each device.  $\nu_{k,t}$, $\lambda_t$ are arbitrary hyper-parameter. 

\subsection{Generalization Error Analysis}
The stability of the sparse-regularized algorithm is analyzed to assess the generalization error bound of LoRA dropout FT by optimizing $\ell_{k,\lambda_t}$. Stability is a well-explored subject in machine learning \cite{charles2018stability, bousquet2002stability} and we adopt a utilized analytical framework, pointwise hypothesis stability (PHS)\cite{fu2023effectiveness}. This approach examines the perturbation of the optimal model when a single training sample is removed. For each device, the entire training dataset is $\mathcal{D}_k$, and the dataset after removing a sample $\boldsymbol{x}_j$ as $\mathcal{D}_k^j=\mathcal{D}_k-\{\boldsymbol{x}_j \}$. It is assumed that $j \sim U(|\mathcal{D}_k|)$, which means the removal is sampled from a uniform distribution. $\boldsymbol{\theta}_{ \ell}(\mathcal{D}_k)$ is defined as the optimal model parameters with respect to (w.r.t.) loss function $\ell$ and dataset $\mathcal{D}_k$.

\begin{definition} \label{definition}
A learning algorithm parameterized by $\boldsymbol{\theta}$ w.r.t. a loss function $\ell$ has PHS $\beta$, if:
\begin{equation}\label{Eq:PHS}
\mathbb{E}_{\mathcal{D}_k, j \sim U(n)}\left|\ell \left( \boldsymbol{\theta}_{ \ell}(\mathcal{D}_k^j), \boldsymbol{x}_j \right)- \ell \left( \boldsymbol{\theta}_{ \ell}(\mathcal{D}_k), \boldsymbol{x}_j \right) \right|  \leq \beta.
\end{equation}
\end{definition}
In \eqref{Eq:PHS}, $\ell \left( \boldsymbol{\theta},\boldsymbol{x}_j \right)$ represents the loss for sample $\boldsymbol{x}_j$ given model parameters $\boldsymbol{\theta}$. According to \cite{lin2024lora}, if the following requirements are met: $\ell$ is $\eta$-Lipschitz, $\boldsymbol{\theta}_{ \ell_{k,\lambda}}(\mathcal{D}_k)$ is close to $\boldsymbol{\theta}_{ \ell_{k,\lambda_t}}(\mathcal{D}_k^j)$, Hessian matrix $\nabla^2{\ell_{k,\lambda_t}}(\boldsymbol{\theta}_{ \ell_{k,\lambda_t}}(\mathcal{D}_k))$ at $\boldsymbol{\theta}_{ \ell_{k,\lambda_t}}(\mathcal{D}_k)$ is positive-semidefinite with a singular value decomposition $U_{k,t} \text{diag}(\Lambda_{k,t})U_{k,t}^{-1}$, and $\Lambda_{k,t,\text{min}}=\text{min}\{\Lambda_{k,t,1}, \cdots, \Lambda_{k,t,m} \}$, then the LoRA dropout algorithm optimizing $ \ell_{k,\lambda_t}$ on $|\mathcal{D}_k|$ has an upper bound of PHS of \footnote{\textcolor{black}{We only require the regularization coefficient to be sufficiently large (i.e., $\lambda_t \geq -\frac{1}{2}\Lambda_{k,t,\min}$), without invoking stronger assumptions about the objective function.} Since uniform stability is stricter and argument stability focuses on parameters, we adopt PHS. Despite the non-convexity of LLM training, which may introduce theoretical gaps, the resulting bound provides meaningful qualitative insights into the client numbers and dropout rates. Tighter bounds under weaker assumptions remain an exciting direction for future work.}
    \begin{equation}
        \begin{aligned}
                &     \mathbb{E}_{\mathcal{D}_k, j \sim U(n)}\left|\ell_{k, \lambda_t} \left( \boldsymbol{\theta}_{ \ell_{k,\lambda_t}}(\mathcal{D}_k^j),\boldsymbol{x}_j \right)- \ell_{k,\lambda_t} \left( \boldsymbol{\theta}_{ \ell_{k,\lambda_t}}(\mathcal{D}_k),\boldsymbol{x}_j \right) \right| \\
                & \leq \frac{2\eta^2}{ (\Lambda_{k,t,\text{min}} +2\lambda_t (2\gamma_{k,t}-\gamma_{k,t}^2) ) \left|\mathcal{D}_k\right|}.
        \end{aligned}
\end{equation}
\begin{theorem} \label{theorem:phs}

    If the conditions in Definition \ref{definition} are met (i.e., \eqref{Eq:PHS}), PHS of LoRA dropout algorithm on each device can be upper-bounded as    
    % then the LoRA dropout algorithm optimizing $ \ell_{k,\lambda_t}$ on $|\mathcal{D}_k|$ has an upper bound of pointwise hypothesis stability of 
    \begin{equation}
        \begin{aligned}
                &     \mathbb{E}_{\mathcal{D}_k, j \sim U(n)}\left|\ell_{k, \lambda_t} \left( \boldsymbol{\theta}_{ \ell_{k,\lambda_t}}(\mathcal{D}_k^j),\boldsymbol{x}_j \right)- \ell_{k,\lambda_t} \left( \boldsymbol{\theta}_{ \ell_{k,\lambda_t}}(\mathcal{D}_k),\boldsymbol{x}_j \right) \right| \\
                & \leq \frac{2\eta^2}{ (\Lambda_{k,t,\text{min}} +2\lambda_t (2\gamma_{k,t}-\gamma_{k,t}^2) ) \left|\mathcal{D}_k\right|}.
        \end{aligned}
\end{equation}
\end{theorem}
\begin{proof}
	Please refer to Appendix \ref{proof_of_theorem}.
\end{proof}

\begin{theorem}\label{theorem:FedLoDrop_PHS} Based on the relation in \eqref{equ:relation_loss}, PHS of LoRA dropout algorithm on the server can be upper-bounded as  
    \begin{equation}
        \begin{aligned}
                &     \mathbb{E}_{\mathcal{D}, {j} \sim U(n)}
                \left| \mathcal{L}^{\prime}_{ \{\nu_{k,t}\}} \left( \boldsymbol{\theta}_{ \mathcal{L}^{\prime}_{ \{\nu_{k,t}\}}}(\mathcal{D}^{\{j\}}), \boldsymbol{x}_{\{j\}} \right) \right.\\
                & \left. \quad\quad\quad\quad\quad\quad - \mathcal{L}^{\prime}_{ \{\nu_{k,t}\}} \left( \boldsymbol{\theta}_{ \mathcal{L}^{\prime}_{ \{\nu_{k,t}\}}}(\mathcal{D}), \boldsymbol{x}_{\{j\}} \right) \right| \\
                & \leq \sum_{k=1}^K \frac{2\eta^2}{ (\Lambda_{k,t,\text{min}} +2\lambda_t (2\gamma_{k,t}-\gamma_{k,t}^2) ) |\mathcal{D}|}.
        \end{aligned}
\end{equation}
\end{theorem}
It follows from Theorem \ref{theorem:FedLoDrop_PHS} that increasing dropout rates (enhancing sparsity) leads to a reduction in the upper bound, indicating that sparser models exhibit greater stability. Consequently, with established stability bounds, the generalization error for the sparse fine-tuned model can be determined.

\begin{lemma}
    For any learning algorithm $\mathbb{M}$ having parameter $W$ and bounded loss function $\ell$ satisfying $0 \leq |\ell(x)-\ell(x^{\prime}) | \leq C, \forall x, x^{\prime}$. If $\mathbb{M}$ has a PHS $\beta$, with probability $1-\delta$, we have:
    \begin{equation}
        R(\mathbb{M},\mathcal{D}) \leq \hat{R}(\mathbb{M},\mathcal{D}) + \sqrt{\frac{C^2+12C|\mathcal{D}|\beta}{2|\mathcal{D}|\delta}},
    \end{equation}
    where $R(\mathbb{M},\mathcal{D})=\sum_k\frac{|\mathcal{D}_k|}{|\mathcal{D}|} \mathbb{E}[\ell_k(\boldsymbol{\theta}, |\mathcal{D}_k|)]$ and $\hat{R}(\mathbb{M},\mathcal{D})=\sum_k\frac{|\mathcal{D}_k|}{|\mathcal{D}|}\ell_k(\boldsymbol{\theta}, |\mathcal{D}_k|)$ denote the generalization risk and empirical risk of algorithm $\mathbb{M}$ running on dataset $\mathcal{D}$, respectively.
\end{lemma}

\begin{theorem}\label{theorem:FedLoD_errbound}
In training round $t$, given a LoRA dropout rate $\gamma_{k,t}$ and sparsity regularization $\lambda_t$, if the conditions in Definition \ref{definition} are met (i.e., \eqref{Eq:PHS}), then for some constant $C$, with probability $1-\delta$,

     \begin{equation} \label{equ:errorbound}
        \begin{aligned}
                    R(M,\mathcal{D}) & \leq \hat{R}(M,\mathcal{D}) \\
                    & \quad+ \sqrt{\frac{C^2+ \sum_{k=1}^K  \frac{24C\eta^2}{\Lambda_{k,t,\text{min}} +2\lambda_t (2\gamma_{k,t}-\gamma_{k,t}^2) }}{2|\mathcal{D}|\delta}}.
        \end{aligned}
    \end{equation}   
\end{theorem}
Theorem \ref{theorem:FedLoD_errbound} elucidates the connection between generalization error and dropout rate. It is shown that increasing the dropout rate decreases the gap between empirical and generalization errors but concurrently increases empirical error. Therefore, an appropriate dropout rate would help balance a trade-off between overfitting and underfitting. \textcolor{black}{By enforcing sparsity more stringently, a larger $\lambda_t$ leads to more compact parameter updates and a lower theoretical generalization error, thereby reducing overfitting risk and improving expected out-of-sample performance.} While the theorem does not explicitly characterize the relationship between generalization error and LoRA rank, it suggests that decreasing the rank-enhancing parameter sparsity-typically correlates with reduced generalization error. This implies that a model with a lower effective rank tends to have better generalization performance \cite{fu2023effectiveness}.

\subsection{Upper Bounding Loss Descent}
However, since $\hat{R}(M,\mathcal{D})$ lacks an analytical expression, the problem is not directly solvable. To address this, upper bounds for $\hat{R}(M,\mathcal{D})$ is derived. Without loss of generality, the following assumptions are made \cite{yang2022transformers,10763424}. We assume a constant batch size and learning rate in the following analysis. 

\begin{assumption}[Lipschitz Continuous Gradient]\label{assumption1}
	$\nabla \mathcal{L} (\Delta \boldsymbol{\theta})$ is Lipschitz continuous with $ \eta > 0$ such that, $		\|  \nabla \mathcal{L}(\Delta \boldsymbol{\theta}_2) - \nabla \mathcal{L}(\Delta \boldsymbol{\theta}_1) \|_F\leq \eta \|\Delta \boldsymbol{\theta}_2-\Delta \boldsymbol{\theta}_1\|_F,$
	% \begin{equation}\label{eqn:smooth}
	% 	\|  \nabla \mathcal{L}(\Delta \boldsymbol{\theta}_2) - \nabla \mathcal{L}(\Delta \boldsymbol{\theta}_1) \|_F\leq \eta \|\Delta \boldsymbol{\theta}_2-\Delta \boldsymbol{\theta}_1\|_F,
	% \end{equation} 
    where $\| \cdot \|_F$ calculates the Frobenius-norm (F-norm).
\end{assumption}

\begin{assumption}[Bounded Gradient Matrix] \label{assumption2} 
The F-norm of gradient matrix is upper bounded, i.e., $\mathbb{E}\left[ \| \boldsymbol{H}(\cdot) \|_F^2 \right] \leq \mathcal{H}^2$.
% \begin{equation}
%     \mathbb{E}\left[ \| \boldsymbol{H}(\cdot) \|_F^2 \right] \leq \mathcal{H}^2.
% \end{equation}
\end{assumption}

% \begin{assumption}[Bounded gradient]\label{assumption3}
% The F-norm of gradients is upper bounded, i.e.,
% \begin{equation}
%     \begin{aligned}
%                  \mathbb{E}\left[ \| \Delta\boldsymbol{A}_{k,t}\|_F^2 \right], \mathbb{E}\left[ \| \Delta\boldsymbol{B}_{k,t}\|_F^2 \right] \leq \mathcal{G}^2, \forall k,t.
%     \end{aligned}
% \end{equation}   
% \end{assumption}

\begin{assumption}[Small Dropout Rate]
\label{assumption3}
The dropout rate of each device is small so that the higher-order terms can be ignored.
% \begin{equation}
%     \| \hat{\boldsymbol{A}}_{k,t}-\boldsymbol{A} \|^2, \| \hat{\boldsymbol{B}}_{k,t}-\boldsymbol{B} \|^2 \rightarrow 0, \quad \forall \boldsymbol{w}.
% \end{equation}
\end{assumption}

\begin{assumption}[Bounded weight]\label{assumption4}
The F-norm of the parameter vector is upper bounded, i.e., $\mathbb{E}\left[ \| {\boldsymbol{A}}_{t} \|^2_F \right],  \mathbb{E}\left[ \| {\boldsymbol{B}}_{t} \|^2_F \right]  \leq \mathcal{G}^2, \forall k,t.$
% \begin{equation}
%          \mathbb{E}\left[ \| {\boldsymbol{A}}_{t} \|^2_F \right],  \mathbb{E}\left[ \| {\boldsymbol{B}}_{t} \|^2_F \right]  \leq \mathcal{G}^2, \forall k,t.
% \end{equation}   
\end{assumption}

\begin{assumption}[Polyak-Łojaciewicz inequality]\label{assumption5}
Optimal loss function value of $\mathcal{L}$ exists, denoted $\mathcal{L}^*$, a constant $\mu \geq 0$ exists, and $\rho =\mathcal{L}(\Delta\boldsymbol{\theta}) - \mathcal{L}^*$ that satisfies $\|\nabla\mathcal{L}(\Delta\boldsymbol{\theta}) \|_F^2  \geq 2\mu\rho.$
% \begin{equation}
%          \|\nabla\mathcal{L}(\Delta\boldsymbol{\theta}) \|_F^2  \geq 2\mu\rho.
% \end{equation}   
\end{assumption}
Specifically, it has been proven in \cite{yang2022transformers,10763424} that the transformer-based LLM possesses a Lipschitz continuous gradient, satisfying Assumption \ref{assumption1}. The PL Assumption \ref{assumption5} is weaker than strong convexity, and usually considered in the analysis of non-convex cases.

% However, the property of Lipschitz smoothness does not necessarily hold when applying LoRA adaptation. Specifically, even if the original function $\mathcal{L}(W)$ is Lipschitz smooth, meaning that the gradient of $\mathcal{L}(W)$ satisfies the Lipschitz continuity condition (as stated in Assumption \ref{assumption1}), this smoothness property is generally lost when the function is expressed in the adapted form $\mathcal{L}(W_0+BA)$.

\begin{lemma}\label{lemma:loragradie}
Similar to \cite{10763424}, the update of $\Delta\boldsymbol{W}_{u,t}$ can be expressed as
% \begin{equation}
%     \begin{aligned}
%         &\Delta\boldsymbol{W}_{u,t}
%         =\boldsymbol{B}_{u,t}\boldsymbol{A}_{u,t}\\
%         & =( {\boldsymbol{B}}_{u,t-1} - \alpha_B \Delta\hat{\boldsymbol{B}}_{u,t-1})( {\boldsymbol{A}}_{u,t-1} - \alpha_A \Delta\hat{\boldsymbol{A}}_{u,t-1})\\
%         % & = \boldsymbol{B}_{u,t-1}\boldsymbol{A}_{u,t-1}-\alpha_A \boldsymbol{B}_{u,t-1}\Delta\hat{\boldsymbol{A}}_{u,t-1}\\
%         % & \quad \quad -\alpha_B \Delta\hat{\boldsymbol{B}}_{u,t-1}\boldsymbol{A}_{u,t-1} +\alpha_B\alpha_A\Delta\hat{\boldsymbol{B}}_{u,t-1}\Delta\hat{\boldsymbol{A}}_{u,t-1}\\
%         & \stackrel{(a)} \approx \boldsymbol{B}_{u,t-1}\boldsymbol{A}_{u,t-1}-\alpha_A \boldsymbol{B}_{u,t-1}\Delta\hat{\boldsymbol{A}}_{u,t-1}\\
%         & \quad \quad -\alpha_B \Delta\hat{\boldsymbol{B}}_{u,t-1}\boldsymbol{A}_{u,t-1}\\
%         & \stackrel{(b)} = \Delta\boldsymbol{W}_{u,t-1} - \alpha(\boldsymbol{B}_{u,t-1}\Delta\hat{\boldsymbol{A}}_{u,t-1}+\Delta\hat{\boldsymbol{B}}_{u,t-1}\boldsymbol{A}_{u,t-1})\\
%         & \stackrel{(c)} = \Delta\boldsymbol{W}_{u,t-1}  - \alpha (\boldsymbol{G}_{u,t-1} - \boldsymbol{J}_{u,t-1}),
%     \end{aligned}
% \end{equation}
\begingroup
\renewcommand{\arraystretch}{0.85}
\setlength{\jot}{-3pt}
\begin{equation}
    \begin{aligned}
        \Delta\boldsymbol{W}_{u,t}
        &= \boldsymbol{B}_{u,t}\boldsymbol{A}_{u,t} \\
        &= ( \boldsymbol{B}_{u,t-1} - \alpha_B \Delta\hat{\boldsymbol{B}}_{u,t-1} )
           ( \boldsymbol{A}_{u,t-1} - \alpha_A \Delta\hat{\boldsymbol{A}}_{u,t-1} ) \\
        &\stackrel{(a)}{\approx} \boldsymbol{B}_{u,t-1}\boldsymbol{A}_{u,t-1}
           - \alpha_A \boldsymbol{B}_{u,t-1}\Delta\hat{\boldsymbol{A}}_{u,t-1} \\
        &\quad - \alpha_B \Delta\hat{\boldsymbol{B}}_{u,t-1}\boldsymbol{A}_{u,t-1} \\
        &\stackrel{(b)}{=} \Delta\boldsymbol{W}_{u,t-1}
           - \alpha ( \boldsymbol{B}_{u,t-1}\Delta\hat{\boldsymbol{A}}_{u,t-1}
           + \Delta\hat{\boldsymbol{B}}_{u,t-1}\boldsymbol{A}_{u,t-1} ) \\
        &\stackrel{(c)}{=} \Delta\boldsymbol{W}_{u,t-1}
           - \alpha ( \boldsymbol{G}_{u,t-1} - \boldsymbol{J}_{u,t-1} ),
    \end{aligned}
\end{equation}
\endgroup
where (a) comes from the slight term of the gradient matrices product caused by the $\alpha_A\alpha_B$, (b) comes from letting $\alpha_A=\alpha_B=\alpha$. In (c), $\boldsymbol{G}_{u,t-1}$ is the desired ideal global gradient matrix, $\boldsymbol{J}_{u,t-1}$ is the global gradient error matrix caused by LoRA dropout, and can be written as $\boldsymbol{J}_{u,t-1}= \boldsymbol{G}_{u,t-1} - \hat{\boldsymbol{G}}_{u,t-1}$, and $\hat{\boldsymbol{G}}_{u,t-1}=\boldsymbol{B}_{u,t-1}\Delta\hat{\boldsymbol{A}}_{u,t-1}+\Delta\hat{\boldsymbol{B}}_{u,t-1}\boldsymbol{A}_{u,t-1}$.
\end{lemma}

\begin{lemma}\label{lemma_layer_error}
    The layer-wise LoRA dropout error $\mathbb{E}\left[\| \boldsymbol{J}_{u,t-1} \|_F^2 \right] $ can be bounded as 
\begin{equation}
    \begin{aligned}
            & \mathbb{E}\left[\| \boldsymbol{J}_{u,t-1} \|_F^2 \right] \\
            & \leq 2 \sum_{k=1}^K  \frac{\left|\mathcal{D}_k\right|}{|\mathcal{D}|} \mathcal{G}^2 \mathcal{H}^2 \mathcal{G}^2 n_2 \gamma_{k,t} + 2 \sum_{k=1}^K  \frac{\left|\mathcal{D}_k\right|}{|\mathcal{D}|} \mathcal{G}^2 \mathcal{H}^2 \mathcal{G}^2 n_1 \gamma_{k,t}\\
            & = 2 (n_1+n_2) \mathcal{H}^2 \mathcal{G}^4 \sum_{k=1}^K  \frac{\left|\mathcal{D}_k\right|}{|\mathcal{D}|} \gamma_{k,t}.
    \end{aligned}
\end{equation}
\begin{proof}
    Please refer to Appendix \ref{proof_of_lemma_layer}.
\end{proof}
\end{lemma}

By leveraging Assumptions \ref{assumption1}-\ref{assumption5} and Lemma \ref{lemma_layer_error}, we derive the expected upper bound of the optimality gap between consecutive training rounds by elucidating how the gradient error matrix $\boldsymbol{J}_{t-1}$ affects the FT procedure.

\begin{theorem} \label{theorem:loss_descent}
    Setting the learning rate as $\alpha=\frac{1}{\eta}$, we get the upper bounding loss descent by gradient norms,
    \begin{equation}
    \setlength{\jot}{0.5\jot} % 压缩多行公式之间的行间距
        \begin{aligned}
                    & \mathbb{E} [\mathcal{L}\left(\Delta\boldsymbol{\theta}_{t}\right) ] - \mathcal{L}\left(\Delta \boldsymbol{\theta}_{t-1}\right) \\
                    & \leq - \frac{\mu\rho}{\eta} + \frac{U^{\prime}(n_1+n_2) \mathcal{H}^2 \mathcal{G}^4 \sum_{k=1}^K  \frac{\left|\mathcal{D}_k\right|}{|\mathcal{D}|} \gamma_{k,t}}{\eta}.
        \end{aligned}
    \end{equation}
\end{theorem}

\begin{proof}
    Please refer to Appendix \ref{proof_of_theorem_conver}.     
\end{proof}

\textcolor{black}{
\subsection{Convergence Analysis} \label{subsect:conver}
In this subsection, we investigate the convergence behavior of FedLoDrop by capturing the learning error caused by the LoRA dropout. }
\textcolor{black}{
Based on \eqref{equ:proof_whole_model}, we have 
\begin{equation}\label{equ:conver}
        \begin{aligned}
                    & \frac{1}{2\eta} \mathbb{E}\left[ \|\nabla\mathcal{L}(\Delta\boldsymbol{\theta}_{t-1}) \|_F^2 \right] \\
                    & \leq \mathbb{E} [\mathcal{L}\left(\Delta\boldsymbol{\theta}_{t-1}\right)  - \mathcal{L}\left(\Delta \boldsymbol{\theta}_{t}\right) ]  + \frac{1}{2\eta}\mathbb{E}\left[\| \boldsymbol{J}_{t-1} \|_F^2 \right].
        \end{aligned}
    \end{equation}
Next, averaging both sides across communication rounds $t=2, 3, \ldots, T+1$, and combining with Lemma \ref{lemma_layer_error}, the average F-norm of global gradient is bounded by
\begin{equation}\label{Equ:conver_final}
\setlength{\jot}{0.5\jot} % 压缩多行公式之间的行间距
    \begin{aligned}
        & \frac{1}{T}\sum_{t=2}^{T+1} \mathbb{E}\left[ \|\nabla\mathcal{L}(\Delta\boldsymbol{\theta}_{t-1}) \|_F^2 \right] 
        \leq \frac{2\eta}{T}\left(\mathcal{L}\left(\Delta\boldsymbol{\theta}_{1}\right) - \mathcal{L}\left(\Delta\boldsymbol{\theta}^{*}\right)\right) \\
        & \quad + \frac{1}{T}\sum_{t=2}^{T+1} 2 (n_1+n_2) U^{\prime} \mathcal{H}^2 \mathcal{G}^4 \sum_{k=1}^K  \frac{\left|\mathcal{D}_k\right|}{|\mathcal{D}|} \gamma_{k,t},
    \end{aligned}
\end{equation}
where $\Delta\boldsymbol{\theta}_{1}$ and $\Delta\boldsymbol{\theta}^{*}$ denote the initialized model and the optimal model, respectively. It can be observed that the average F-norm of global gradient matrix is bounded by two terms: the first captures the discrepancy between $\Delta\boldsymbol{\theta}_{1}$ and $\Delta\boldsymbol{\theta}^{*}$, while the second reflects the optimality gap introduced by LoRA dropout. As the number of communication rounds tends to infinity, the right-hand side of \eqref{Equ:conver_final} converges to zero. Furthermore, the convergence rate decreases with an increasing dropout rate, since FedLoDrop reduces per-round communication overhead at the cost of requiring more rounds to achieve convergence. 
}

\section{Edge Implementation of FedLoDrop}
\subsection{System Model}
\subsubsection{Network Model}
A single-cell network implements the FedLoDrop framework, as shown in Fig. \ref{fig:system}. The system includes one edge server with a single-antenna access point (AP) and $K$ single-antenna edge devices that cooperatively perform FL via wireless links. Channels are frequency non-selective, static within each round but vary across rounds. Using orthogonal frequency division multiplexing (OFDM), different users are allocated distinct subcarriers to avoid intra-cell interference. Each device uses separate subcarriers for downloading and uploading. The server obtains channel state information (CSI) of channels of all devices by using effective channel estimation methods (see, e.g., \cite{9470915}).

\subsubsection{Latency and Energy Consumption Models}
Consider an arbitrary communication round $t$ and device $k$. We ignore latency and energy consumption in the sub-adapter generation and aggregation steps. The other three steps are modeled as follows:
\begin{itemize}
     \item \emph{Downloading Step}: Denote the latency of device $k$ on subcarrier $s$ to download its assigned sub-adapter as $ T_{k,t,s}^{\rm {com,dl}}$. If subcarrier $s$ is not allocated to $k$, i.e., $z_{k,t,s}=0$, $ T_{k,t,s}^{\rm {com,dl}}=0$. Otherwise,
\begin{equation}
    T_{k,t,s}^{\rm {com,dl}}=\dfrac{\hat{M}_{k,t,s} Q}{R_{k,t,s}^{{\rm dl}}}, \; \forall z_{k,t,s}=1,
\end{equation}
where $Q$ is the quantization bits used for one parameter, $\hat{M}_{k,t,s}$ is the number of parameters uploaded by device $k$ on subcarrier $s$. $R_{k,t,s}^{\rm dl}$ is the channel capacity, given as
\begin{equation}\label{equ:r_downlink}
    \begin{aligned}
    R_{k,t,s}^{\rm dl}
    &= B\log_2 \left(1+\frac{|h_{k,t,s}^{\rm dl}|^2 P_{k,t,s}^{\rm{{com,dl}}}}{\sigma^2} \right),
    \end{aligned}
\end{equation}
where $\sigma^2$ is the power of additive white Gaussian noise, $h_{k,t,s}^{\rm dl}$ is the downlink channel gain, and $P_{k,t,s}^{\rm{com,dl}}$ is the downlink transmit power of device $k$ on subcarrier $s$, respectively. Then, the overall uploading latency of device $k$ is decided
by the slowest subcarrier:
\begin{equation}\label{eqn:communicationdl_time}
     T_{k,t}^{\rm {com,dl}}=\max_s \; T_{k,t,s}^{\rm {com,dl}}.
\end{equation}

The energy consumption in this step is to receive the model, which is included in the circuit energy consumption $\xi_k$.

\item \emph{Local Tuning Step}: The latency of device $k$ is given by 
\begin{equation}\label{eqn:computation_time}
    T_{k,t}^{\rm {cmp }}=\dfrac{{C}_{k,t}|\mathcal{D}_k|}{f_{k,t}},
\end{equation}
where $f_{k,t}$ (in cycle/s) is the CPU/GPU frequency of device $k$, ${C}_{k,t}$ is the number of processor operations to update the subnet using one data sample. 
Then, following \cite{you2017energy}, the energy consumption is given by
\begin{equation}\label{eqn:computation_energy}
     E_{k,t}^{\rm {cmp}}=\Omega_kT_{k,t}^{\rm {cmp }} (f_{k,t})^3,
\end{equation}
where $\Omega_k$ is a constant characterizing the local computation performance of the processor on device $k$.

\item \emph{Uploading Step}: Denote the latency of device $k$ on subcarrier $s$ to upload its assigned sub-adapter as $ T_{k,t,s}^{\rm {com,ul}}$. If subcarrier $s$ is not allocated to $k$, i.e., $z_{k,t,s}=0$, $ T_{k,t,s}^{\rm {com,ul}}=0$. Otherwise,
\begin{equation}
    T_{k,t,s}^{\rm {com,ul}}=\dfrac{\hat{M}_{k,t,s} Q}{R_{k,t,s}^{\rm ul}}, \; \forall z_{k,t,s}=1,
\end{equation}
where $R_{k,t,s}^{\rm ul}$ is the channel capacity, given by
\begin{equation} \label{equ:r_k,t}
    R_{k,t,s}^{\rm ul}= B\log_2 \left(1+\frac{|h_{k,t,s}^{\rm ul}|^2P_{k,t,s}^{\rm{com,ul}}}{\sigma^2} \right).
\end{equation}
In \eqref{equ:r_k,t}, $h_{k,t,s}^{\rm ul}$ is the uplink channel gain,  and other notations follow that in \eqref{equ:r_downlink}. Based on \eqref{equ:r_k,t}, the uplink transmit power is derived as
\begin{equation}\label{Eq:TransmitPower}
    P_{k,t,s}^{\rm{com,ul}}=\left(2^{\frac{R_{k,t,s}^{\rm ul}}{B}}-1 \right)\frac{\sigma^2}{|h_{k,t,s}^{\rm ul}|^2}.
\end{equation}
Then, the overall uploading latency of device $k$ is decided by the slowest subcarrier:
\begin{equation}\label{eqn:communicationul_time}
     T_{k,t}^{\rm {com,ul}}=\max_s \; T_{k,t,s}^{\rm {com,ul}}.
\end{equation}

If subcarrier $s$ is not allocated to $k$, i.e., $z_{k,t,s}=0$, $ E_{k,t,s}^{\rm{com,ul}}=0$. Otherwise, 
$E_{k,t,s}^{\rm{com,ul}} =z_{k,t,s} P_{k,t,s}^{\rm {com,ul}}T_{k,t,s}^{\rm {com,ul}}.$
% \begin{equation}
%      \begin{aligned}
%          E_{k,t,s}^{\rm{com,ul}}
%          & =z_{k,t,s} P_{k,t,s}^{\rm {com,ul}}T_{k,t,s}^{\rm {com,ul}}.%\\
%          %& =\left(2^{\frac{M_kQ}{T_{k,t}^{\text {com,ul}}\rho_{k,t}B}}-1 \right)\frac{N_0}{h_{k,t}^{ul}}T_{k,t}^\text{{com,ul}},
%      \end{aligned}
% \end{equation}

The total uploading energy consumption of device $k$ is the sum of uploading energy consumption over all subcarriers, given by
\begin{equation} \label{eqn:communicationul_energy}
    \begin{aligned}
        E_{k,t}^{\rm{com,ul}}
         & =  \sum_{s=1}^S E_{k,t,s}^{\rm{com,ul}}.
         % \\
         % & =  \sum_{s=1}^S \frac{z_{k,t,s}\left(2^{\frac{R_{k,t,s}^{\rm ul}}{B}}-1 \right) \sigma^2\hat{M}_{k,t,s} Q} {|h_{k,t,s}^{\rm ul}|^2 R_{k,t,s}^{\rm ul}}
    \end{aligned}
\end{equation}

In summary, the overall latency of device $k$ in this round is derived as
\begin{equation}\label{Eq:DeviceLatency}
    T_{k,t}=T_{k,t}^{\rm {com,dl}}+T_{k,t}^{\rm {cmp}}+T_{k,t}^{\rm {com,ul}},\; \forall k,
\end{equation}
where $T_{k,t}^{\rm {com,dl}}$ is the downlink communication latency defined in \eqref{eqn:communicationdl_time},
$T_{k,t}^{\rm {cmp}}$ is the computation latency defined in \eqref{eqn:computation_time}, and $T_{k,t}^{\rm {com,ul}}$ is the uplink communication latency defined in \eqref{eqn:communicationul_time}.\\
The total energy consumption of device $k$ in this round is 
\begin{equation}\label{Eq:DeviceEnergy}
    E_{k,t}=E_{k,t}^{\rm {com,ul}}+E_{k,t}^{\rm{cmp}}+\xi_k, \; \forall k,
\end{equation}
where $E_{k,t}^{\rm{cmp}}$ is the computation energy consumption in \eqref{eqn:computation_energy}, $E_{k,t}^{\rm{com,ul}}$ is the communication energy consumption in \eqref{eqn:communicationul_energy}, and $\xi_k$ is the circuit energy consumption for global model reception.

\end{itemize}

\subsection{Problem Formulation}
Based on the error bound provided in \eqref{equ:errorbound}, this subsection formulates an optimization problem to jointly design the adaptive dropout rate and resource allocation, aiming to minimize this bound. Since probability $\delta$ and other parameters, i.e., $C$ and $\eta$ are constants, minimizing the right side of \eqref{equ:errorbound} is equivalent to
% \begin{figure*}
    \begin{equation} \label{equ:optim}
        \begin{aligned}
            \min\limits_{\{{\gamma}_{k,t},
				{z}_{k,t,s},
                    \hat{M}_{k,t,s}\}}\;\; 
    & \frac{U^{\prime}(n_1+n_2) \mathcal{H}^2 \mathcal{G}^4 \sum_{k=1}^K  \frac{\left|\mathcal{D}_k\right|}{|\mathcal{D}|} \gamma_{k,t}}{\eta} \\
    & + \sqrt{\frac{\sum_{k=1}^K  \frac{12C\eta^2}{\Lambda_{k,t,\text{min}} +2\lambda_t (2\gamma_{k,t}-\gamma_{k,t}^2) }}{|\mathcal{D}|\delta}}.
        \end{aligned}
\end{equation}
% \end{figure*}

Due to the limited network resources, there are several constraints on the latency, energy consumption, subcarrier allocation, and LoRA dropout rate in each round $t$, as elaborated below.
\subsubsection{Per-Round Latency Constraint}
The latency of each device should not exceed the maximum permitted latency $T_0$ to complete this round. Based on the devices' latency $\{T_{k,t}\}$ derived in \eqref{Eq:DeviceLatency}, the latency constraint is given by 
%\begin{equation} \label{Eq:con_lat}
% T_{k,t}\leq T_0, \; 1\leq k \; \leq K,
%\end{equation}
%where $T_0$ is the maximum allowed latency to complete this round.  By substituting $T_{k,t}$ in \eqref{Eq:DeviceLatency}, the constraint in \eqref{Eq:con_lat} is  derived as
\begin{equation} \label{Eq:con_lat}
 T_{k,t}^{\rm {com,dl}}+T_{k,t}^{\rm {cmp}}+T_{k,t}^{\rm {com,ul}} \leq T_0, \; \forall k.
\end{equation}
%where $T_0$ is the maximum permitted latency to complete this round. 
By substituting $T_{k,t}^{\rm {com,dl}}$ given in \eqref{eqn:communicationdl_time},
$T_{k,t}^{\rm {cmp}}$ given in \eqref{eqn:computation_time}, and $T_{k,t}^{\rm {com,ul}}$ given in \eqref{eqn:communicationul_time} into the latency constraint, we have
\begin{equation}
    \mathcal{C}_1: \; \max_s \; T_{k,t,s}^{\rm {com,dl}}+T_{k,t}^{\rm {cmp}}+ \max_s \; T_{k,t,s}^{\rm {com,ul}} \leq T_0, \; \forall z_{k,t,s}=1.
\end{equation}
% where $\hat{M}_{k,t} = (1-\gamma_{k,t})M_{\rm ori} $ is the number of parameters of the sub-adapter and $\hat{C}_{k,t}=C_{\text{fr}}+(2\gamma_{k,t}-\gamma_{k,t}^2)C_{\text{BA}} + (1-\gamma_{k,t})C_{\text{ba}} $ is the number of processor operations on one data sample. 

\subsubsection{Energy Consumption Constraint}
The total energy consumption of each device
should be no larger than its energy budget $E_{k,0}$. Based on the energy consumption $\{E_{k,t}\}$ of devices given in \eqref{Eq:DeviceEnergy}, the energy consumption constraint is given by 
%\begin{equation} \label{Eq:con_enr}
%E_{k,t}\leq E_{k,0}.
%\end{equation}
%$E_{k,0}$ is the maximal energy consumption for an arbitrary round and can be different over different rounds. By substituting $E_{k,t}$ in \eqref{Eq:DeviceEnergy}, the constraints in \eqref{Eq:con_enr} can be derived as
\begin{equation} \label{Eq:con_enr}
    \mathcal{C}_2:\; E_{k,t}^{\rm {com,ul}}+E_{k,t}^{\rm{cmp}}+\xi_k \leq E_{k,0}, \; \forall k,
\end{equation}
where $E_{k,t}^{\rm{cmp}}$ is defined in \eqref{eqn:computation_energy} and $E_{k,t}^{\rm{com,ul}}$ is defined in  \eqref{eqn:communicationul_energy}.
% the energy consumption constraint is derived as
% \begin{equation}
%     \mathcal{C}_2:\;  E_{k,t,s}^{\rm {com,ul}}+E_{k,t}^{\rm{cmp}}+\xi_k \leq E_{k,0}, \forall z_{k,t,s}=1.
% \end{equation}

\subsubsection{Subcarrier Assignment Constraint}
Each subcarrier can be allocated to one worker:
\begin{equation}
\mathcal{C}_3:\; \left\{\begin{array}{l}
z_{k, t, s} \in \{0,1\}, \; \forall k,\\
\sum_{k=1}^K  z_{k, t, s}= 1.
\end{array}\right.
\end{equation}
where $\sum_{k=1}^K  z_{k, t, s}= 1$ represents that the subcarrier $s$ is allocated to device $k$. 

\subsubsection{Parameter Constraint}
For each device, the total uploaded
number of parameters on all subcarriers should be no smaller than its allocated subnet size:
\begin{equation}
    \mathcal{C}_4:\; \sum_{s=1}^S z_{k, t, s} \hat{M}_{k,t,s} \geq \hat{M}_{k,t}, \; \forall k,
\end{equation}
where $\hat{M}_{k,t}=(1-\gamma_{k,t})M$.
% , and $M$ is the number of parameters in the original work.

\subsubsection{Dropout Rate Constraint}
Based on the definition, the dropout rate of each device in each communication round should be limited between 0 and 1, namely,
\begin{equation}
\mathcal{C}_5:\; 0 \leq \gamma_{k,t} < 1, \; \forall k.
\end{equation}

Under the constraints above, the optimization problem is formulated as 
\eqref{equ:P_1}.
\begin{figure*}
    \begin{equation}\label{equ:P_1}
\mathscr P_1: \quad
	\begin{array}{cl}
		\underset
		{
			\underset
			{
				k \in \mathcal{K} 
			}
			{
				\left \{
				{\gamma}_{k,t},
				{z}_{k,t,s},
                    \hat{M}_{k,t,s}
				\right \}
			}
		}
		{
			\min
		}
		&	
		\frac{U^{\prime}(n_1+n_2) \mathcal{H}^2 \mathcal{G}^4 \sum_{k=1}^K  \frac{\left|\mathcal{D}_k\right|}{|\mathcal{D}|} \gamma_{k,t}}{\eta} + \sqrt{\frac{ \sum_{k=1}^K  \frac{12C\eta^2}{\Lambda_{k,t,\text{min}} +2\lambda_t (2\gamma_{k,t}-\gamma_{k,t}^2) }}{|\mathcal{D}|\delta}},
		\\
		\text{s.t.}
		& \mathcal{C}_1 \sim
         \mathcal{C}_5.
	\end{array}
\end{equation}\vspace{-0.5cm}
\end{figure*}
%\vspace{-0.5cm}
This is a mixed-integer non-linear problem (MINLP) and thus NP-hard. In the sequel, an arbitrary round is considered and the notation $t$ is omitted for simplicity.

\section{Optimized FedLoDrop Implementation}
This section first presents an algorithm based on the B\&B framework to jointly optimize $\mathscr{P}_1$, which is a well-established approach for handling mixed binary-continuous optimization problems. To reduce the complexity, a P-SCA-based algorithm is then proposed. In the context of the OFDM systems considered in this work, the problem is further complicated by the binary subcarrier allocation among devices, along with the associated inter-subcarrier parameter allocation for each device. \textcolor{black}{The resource allocation is dynamically adjusted during training, making it adaptive to changing conditions and highly practical for real-world environments. To perform this optimization, the system requires the CSI, data resource profile, and computational capacity of the devices.}

\subsection{An Equivalent Problem Formulation}
To handle the integer variables, $\hat{M}_{k,s}$ are relaxed to be continuous for computational simplicity. These values are rounded to integers for implementation, with negligible loss due to the typically large magnitudes \cite{9470915}. For the binary variables $z_{k,s}$, the B\&B-based algorithm is employed. The core idea of B\&B is to iteratively partition the feasible region and prune suboptimal subspaces through bound comparisons. Initially, the root node (i.e., the original problem) is relaxed by ignoring integrality constraints, in order to compute a bounding value.
% If the relaxed solution violates integrality, branching divides the problem into subproblems by restricting discrete variables. Each subproblem is then solved as a relaxed program, updating the global bound and pruning branches where bounds cannot improve the best-known feasible solution.
To the end, the following variables are
used to transform $\mathscr P_1$ into a convex problem.

\begin{equation}\label{equ:aul_val}
\; \left\{\begin{array}{l}
\Tilde{\gamma}_{k} = 1 - \gamma_k,\\
\Tilde{M}_{k,s}=z_{k,s}\hat{M}_{k,s}.
\end{array}\right.
\end{equation}
Thus, for a given $z_{k,s}$, the primal problem of $\mathscr P_1$ can be written as
% \begin{figure*}
% \centering
% \begin{minipage}{0.9\textwidth}
% \begin{equation}\label{equ:P_2}
% \begin{aligned}
% \mathscr{P}_2 \quad \min_{\substack{\{\tilde{\gamma}_{k}\}, \{\tilde{M}_{k,s}\},  \{T_{k}^{\mathrm{com,ul}}\}, \{T_{k}^{\mathrm{com,dl}}\}}} 
% & \frac{I}{\eta} \sum_{k=1}^K \frac{|\mathcal{D}_k|}{|\mathcal{D}|} (1 - \tilde{\gamma}_{k}) 
% + \frac{V\eta}{\sqrt{|\mathcal{D}|}} \sqrt{ \sum_{k=1}^K \frac{1}{-\tilde{\gamma}_{k}^2 + a_k }},
% \end{aligned}
% \end{equation}
% \vspace{-1.5\baselineskip}
% \begin{align*}
% \text{subject to} \quad
% & T_{k}^{\mathrm{com,dl}} + \frac{C_k |\mathcal{D}_k|}{f_{k}} + T_{k}^{\mathrm{com,ul}} \leq T_0, \\
% & \sum_{s=1}^S \frac{\tilde{M}_{k,s} \left(2^{\frac{R_{k,s}^{\mathrm{ul}}}{B}} - 1 \right) \sigma^2 Q}{|h_{k,s}^{\mathrm{ul}}|^2 R_{k,s}^{\mathrm{ul}}} + C_k |\mathcal{D}_k| \Omega_k f_k^2 + \xi_k \leq E_{k,0}, \\
% & \sum_{s=1}^S \tilde{M}_{k,s} \geq \tilde{\gamma}_k M, \quad \forall k, \\
% & 0 \leq \tilde{\gamma}_{k} < 1, \quad \forall k, \\
% & T_{k}^{\mathrm{com,dl}} \geq \frac{\tilde{M}_{k,s} Q}{R_{k,s}^{\mathrm{dl}}}, \quad \forall (k,s), \\
% & T_{k}^{\mathrm{com,ul}} \geq \frac{\tilde{M}_{k,s} Q}{R_{k,s}^{\mathrm{ul}}}, \quad \forall (k,s).
% \end{align*}
% \end{minipage}
% \end{figure*}

\textcolor{black}{
\begin{equation*}
\label{equ:P_2}
\begin{aligned}
\mathscr{P}_2:\quad 
\min_{\substack{
    \{\Tilde{\gamma}_k\},\; \{\Tilde{M}_{k,s}\}, \\
    \{T_k^{\rm com,ul}\},\; \{T_k^{\rm com,dl}\}
}} 
\quad & \frac{I}{\eta} \sum_{k=1}^K \frac{|\mathcal{D}_k|}{|\mathcal{D}|} (1 - \Tilde{\gamma}_k) \\
&+ \frac{V\eta}{\sqrt{|\mathcal{D}|}} \sqrt{ \sum_{k=1}^K \frac{1}{ -\Tilde{\gamma}_k^2 + a_k } },
\end{aligned}
\end{equation*}
\begin{align*}
\text{s.t.} \quad
& T_k^{\rm com,dl} + \frac{C_k |\mathcal{D}_k|}{f_k} + T_k^{\rm com,ul} \leq T_0, \\
& \sum_{s=1}^S \frac{ \Tilde{M}_{k,s} \left( 2^{\frac{R_{k,s}^{\rm ul}}{B}} - 1 \right) \sigma^2 Q }{ |h_{k,s}^{\rm ul}|^2 R_{k,s}^{\rm ul} } + C_k |\mathcal{D}_k| \Omega_k f_k^2 + \xi_k \leq E_{k,0}, \\
& \sum_{s=1}^S \Tilde{M}_{k,s} \geq \Tilde{\gamma}_k M, \quad \forall k, \\
& 0 \leq \Tilde{\gamma}_k < 1, \quad \forall k, \\
& T_k^{\rm com,dl} \geq \frac{ \Tilde{M}_{k,s} Q }{ R_{k,s}^{\rm dl} }, \quad \forall (k,s), \\
& T_k^{\rm com,ul} \geq \frac{ \Tilde{M}_{k,s} Q }{ R_{k,s}^{\rm ul} }, \quad \forall (k,s),
\end{align*}
\noindent where $I = U'(n_1 + n_2) \mathcal{H}^2 \mathcal{G}^4$, $V = \sqrt{\frac{6C}{\delta \lambda}}$, and $a_k = \frac{2\lambda + \Lambda_{k,\min}}{2\lambda}$.
}

\begin{lemma}\label{lemma_convex}
    Problem $\mathscr P_2$ is convex.
\begin{proof}
    Please refer to Appendix \ref{proof_of_lemma_convex}.
\end{proof}
\end{lemma}

\subsection{Global Optimal Solution}
As sub-problem $\mathscr P_2$ is convex, its optimal solution is obtained by the primal-dual method, which is summarized in Algorithm \ref{algo:global}. 

\begin{remark}
   Computational Complexity Analysis: In our B\&B-based method, the problem is solved iteratively. At each iteration, the primal problem, which has a complexity of $\mathcal{O}(K^2 S)$, is solved using the primal-dual method. Let $I_{\text{iter}}$ be the number of iterations for the master problem. Thus, the total complexity is $\mathcal{O}(I_{\text{iter}} \cdot K^2 S)$, reaching up to $\mathcal{O}(2^{KS} \cdot K^2 S)$ in the worst case. 
   
   % In the proposed B\&B-based method, the problem is solved iteratively, where at each iteration the primal problem is addressed using the primal-dual method due to its convexity. Consequently, the overall computational complexity involves solving both the master and primal problems. Let $I_{\text{iter}}$ denote the number of iterations required to solve the master problem. At each iteration, let ${O}_a$ represent the computational complexity order of solving the primal problem, which is $\mathcal{O}(K^2 S)$, where $K$ is the number of workers and $S$ is the number of subcarriers. Therefore, the total computational complexity of the B\&B-based method is on the order of $\mathcal{O}(I_{\text{iter}} \cdot O_a)$, which can reach up to $\mathcal{O}(2^{KS} \cdot K^2 S)$ in the worst case.
\end{remark}

\vspace{-0.5cm}
\begin{algorithm}[]\label{algo:global}
	\caption{B\&B-based Primal-dual Algorithm}\label{Alg:Solution}
	\LinesNumbered
	\KwIn{ \{$R_k^{\rm{dl}}$\},  \{$R_k^{\rm{ul}}$\}  \{$P_k^{\rm{com,ul}}$\}, \{$f_k$\}, $T_0$, \{$E_{k,0}$\}.} 

     \textbf{Initialize} \{$\zeta_s^{(0)}$\},  \{$\iota_k^{(0)}$\},\{$\chi_k^{(0)}$\}, \{$\kappa_k^{(0)}$\}, 
     \{$\phi_k^{(0)}$\}, \{$\psi_k^{(0)}$\}, the step sizes \{$\eta_{\zeta s}$\}, \{$\eta_{\iota k}$\}, \{$\eta_{\chi k}$\}, \{$\eta_{\kappa k}$\}, \{$\eta_{\phi k}$\}, \{$\eta_{\psi k}$\}, and $i=0$. \\
    \textbf{Loop}.\\   
    Update the multipliers as
    $$
    \zeta_s^{(i+1)}= \max \left\{ \zeta_s^{(i)} + \eta_{\zeta_s} \frac{\partial \mathcal{L}_{\mathrm{P2}} }{\partial \zeta_s}, \; 0 \right\}, 
    $$
    $$
    \iota_k^{(i+1)}= \max \left\{ \iota_k^{(i)} +  \eta_{\iota k} \frac{\partial \mathcal{L}_{\mathrm{P2}} }{\partial \iota_k},\; 0 \right\}, \forall k,
    $$    
    \[
    	\begin{split}
	\chi_k^{(i+1)} 
       = ~& \max \left\{ \chi_k^{(i)} + \eta_{\chi k}\frac{\partial \mathcal{L}_{\mathrm{P2}} }{\partial \chi_k}, \; 0 \right\}, \forall k,
	    \end{split}
    \]   
    \\
    Calculate \{$\Tilde{\gamma}_{k}$\} and \{$ \Tilde{M}_{k,s}$\}.\\ 
    Get \{$\gamma_{k}$\} and \{$ \hat{M}_{k,s}$\} using \eqref{equ:aul_val}. \\
    $i = i + 1$.\\
    \textbf{Until Convergence}.\\
    \KwOut{\{$\gamma_{k}$\} and \{$ \hat{M}_{k,s}$\}.}
\end{algorithm}

\subsection{Low-complexity suboptimal solution}
Although the proposed B\&B-based algorithm obtains the globally optimal solution, it may incur high computational complexity due to the iterative partitioning and search for the best-known feasible solution \cite{li2006nonlinear, 9140412}. To address this limitation, a low-complexity P-SCA-based algorithm is proposed as an alternative, aiming to strike a balance between computational efficiency and system performance.

Specifically, the non-convex constraint (C3) is converted to an equivalent form with continuous variables as follows:
\begin{equation}
     \mathcal{C}_3 \rightarrow 
    \begin{cases}
   \mathcal{C}_6: z_{k,s} (1 - z_{k,s}) \leq 0, \\
    \mathcal{C}_7: 0 \leq z_{k,s} \leq 1.
    \end{cases}
\end{equation}
By applying the penalty method, objective function of $\mathscr P_2$ is transformed to $\frac{I}{\eta}\sum_{k=1}^K \frac{  \left|\mathcal{D}_k\right|}{|\mathcal{D}|} (1-\Tilde{\gamma}_{k}) +  \frac{V\eta}{\sqrt{|\mathcal{D}|}} \sqrt{ \sum_{k=1}^K  \frac{1}{-\Tilde{\gamma}_{k}^2 + a_k }} + \tau \sum_{k=1}^K\sum_{s=1}^S (z_{k,s}-z_{k,s}^2)$. Approximating the objective function by using the first-order Taylor expansion. $z_{k,s}-z_{k,s}^2$ is rewritten as a difference-of-convex (DC) function shown as $z_{k,s}-z_{k,s}^2 = \Psi_1(z_{k,s}) - \Psi_2(z_{k,s})$, where $\Psi_1(z_{k,s})=z_{k,s}$, $\Psi_2(z_{k,s})=z_{k,s}^2$. Thus, the convex lower bound of $\Psi_2(z_{k,s})$ can be obtained by the first-order Taylor expansion, i.e., 
% \begin{equation}
%     \begin{aligned}
%              \bar{\Psi}_2(z_{k,s}) 
%              & = {\Psi}_2(\bar{z}_{k,s}) + \nabla {\Psi}_2(\bar{z}_{k,s}) (z_{k,s} - \bar{z}_{k,s})\\
%              & = 2 \bar{z}_{k,s}z_{k,s} - (\bar{z}_{k,s})^2,
%     \end{aligned}
% \end{equation}
\begingroup
\renewcommand{\arraystretch}{0.85}
\setlength{\jot}{-2pt}
\begin{equation}
    \begin{aligned}
        \bar{\Psi}_2(z_{k,s})
        &= \Psi_2(\bar{z}_{k,s}) + \nabla\Psi_2(\bar{z}_{k,s})\,(z_{k,s} - \bar{z}_{k,s}) \\
        &= 2\,\bar{z}_{k,s} z_{k,s} - (\bar{z}_{k,s})^2,
    \end{aligned}
\end{equation}
\endgroup
where $\bar{z}_{k,s}$ is represents a feasible solution to Problem $\mathscr P_2$. With $ \bar{\Psi}_2(z_{k,s}) $, $z_{k,s}-z_{k,s}^2 $ is approximated by its convex upper bound denoted as $z_{k,s}-z_{k,s}^2 \leq \Psi_1(z_{k,s}) - \bar{\Psi}_2(z_{k,s}) $.
% = z_{k,s}-2 \bar{z}_{k,s}z_{k,s} + (\bar{z}_{k,s})^2
An upper bound of the problem can be found by solving 
\begin{equation*} \label{equ:P_3}
\begin{aligned}
& \mathscr P_3    \min_{\substack{\{\Tilde{\gamma}_{k}\}, \{\Tilde{M}_{k,s}\}, \{{z}_{k,s}\}\\ \{T_{k}^{\rm {com,ul}}\},\{T_{k}^{\rm {com,dl}}\}}} 
         \frac{I}{\eta}\sum_{k=1}^K \frac{  \left|\mathcal{D}_k\right|}{|\mathcal{D}|} (1-\Tilde{\gamma}_{k}) \\
        &+  \frac{V\eta}{\sqrt{|\mathcal{D}|}} \sqrt{ \sum_{k=1}^K  \frac{1}{-\Tilde{\gamma}_{k}^2 + a_k }} + \tau \sum_{k=1}^K\sum_{s=1}^S (1 -2 \bar{z}_{k,s})z_{k,s} + (\bar{z}_{k,s})^2.
\end{aligned}
\end{equation*}
Note that problem $\mathscr P_3$ is convex and can be solved by using the primal-dual method as well.
% standard convex numerical solvers, such as CVX.

\begin{remark}
   Computational Complexity Analysis: solving problem $\mathscr P_3$ is with the complexity order of $\mathcal{O}(I_{\text{iter}}^{\prime} \cdot K^2 S)$, where $I_{\text{iter}}^{\prime}$ represents the iteration numbers.
\end{remark}

\section{Simulation Results}
We evaluate the superiority of the proposed FedLoDrop framework and compare it with other baseline methods on two major NLP tasks: natural language understanding (NLU) and question answering (QA). We first fine-tune two LLMs including RoBERTa large (355M) and LLaMA (7B) on two benchmarks including GLUE and MMLU. Following the configurations in the original LoRA paper \cite{hu2021lora}, the LoRA modules are applied to the self-attention layers. For RoBERTa large, LoRA rank is set as $4$ in a three-client cross-silo federated setting, and $r = 16$ with $10$ devices for LLaMA in \cite{sun2024improving,singhal2024exact}. \textcolor{black}{The local batch size, epoch and learning rate are set to be 64, 1, and 3e-4, respectively. The AdamW optimizer is adopted during the training.} The downlink channel gains of devices are assumed to be Rayleigh fading. The path loss is $10^{-3}$. For data distribution among clients, a common method of randomly sampling non-IID data for each client is employed, as implemented in \cite{zhang2024towards,lai2022fedscale}. \textcolor{black}{Experiments are performed using 4 NVIDIA RTX 4090 (24 GB each) GPUs.}

\subsection{Effects of Dropout Rate}
\label{subsect:fix_drop}

% \vspace{-0.6cm}
\begin{table*}[htbp]
\centering
\caption{Glue Benchmark}
\label{table:glue}
\begin{tabular}{|c|c|c|c|c|c|c|c|}
\hline
Method                                                                                             & CoLa (Mcc)                   & RTE                          & MRPC                         & SST-2                        & QNLI                         & STS-B (Corr)                 & Avg                          \\ \hline
FedIT in \cite{zhang2024towards}                                                   & 62.08                        & 77.26                        & \textbf{90.44}               & 95.99                        & 94.86                        & \textbf{91.54}               & 85.36                        \\ \hline
\begin{tabular}[c]{@{}c@{}}FedLoDrop\\ (0.2)\end{tabular}                                          & \textbf{63.56}               & 78.70                        & 89.46                        & \textbf{96.67}               & 94.95                        & 91.34                        & 85.78                        \\ \hline
\begin{tabular}[c]{@{}c@{}}FedLoDrop\\ (0.3)\end{tabular}                                          & 62.84                        & \textbf{81.6}                & 88.48                        & 96.56                        & \textbf{95.17}               & 90.84                        & \textbf{85.92}               \\ \hline
FFA in \cite{sun2024improving}                                                   & 58.32                        & 70.76                        & 87.49                        & 95.59                        & 93.85                        & 90.24                        & 82.71                        \\ \hline
\begin{tabular}[c]{@{}c@{}}FFA with Bernoulli Dropout\\ (0.2)\end{tabular}                         & 62.37                        & 74.73                        & 83.33                        & 95.76                        & 94.55                        & 89.12                        & 83.31                        \\ \hline
\begin{tabular}[c]{@{}c@{}}FFA with Bernoulli Dropout\\ (0.3)\end{tabular}                         & 61.83                        & 69.31                        & 78.18                        & 96.22                        & 94.29                        & 88.84                        & 81.45                        \\ \hline
{\color[HTML]{000000} \begin{tabular}[c]{@{}c@{}}FedIT with Gaussian Dropout\\ (0.2)\end{tabular}} & {\color[HTML]{000000} 62.08} & {\color[HTML]{000000} 77.98} & {\color[HTML]{000000} 90.20} & {\color[HTML]{000000} 96.22} & {\color[HTML]{000000} 94.25} & {\color[HTML]{000000} 91.72} & {\color[HTML]{000000} 85.41} \\ \hline
{\color[HTML]{000000} \begin{tabular}[c]{@{}c@{}}FedIT with Gaussian Dropout\\ (0.3)\end{tabular}} & {\color[HTML]{000000} 61.82} & {\color[HTML]{000000} 77.98} & {\color[HTML]{000000} 89.22} & {\color[HTML]{000000} 95.41} & {\color[HTML]{000000} 93.85} & {\color[HTML]{000000} 91.47} & {\color[HTML]{000000} 84.96} \\ \hline
\end{tabular}\vspace{-0.6cm}
\end{table*}

\begin{figure*}[h] 
    \begin{center}
    \subfigure[Accuracy v.s. dropout rate on RoBERTa large]
    {
    \label{Fig:FedLoDrop_rate_Bert}
    \includegraphics[width=0.4\textwidth]{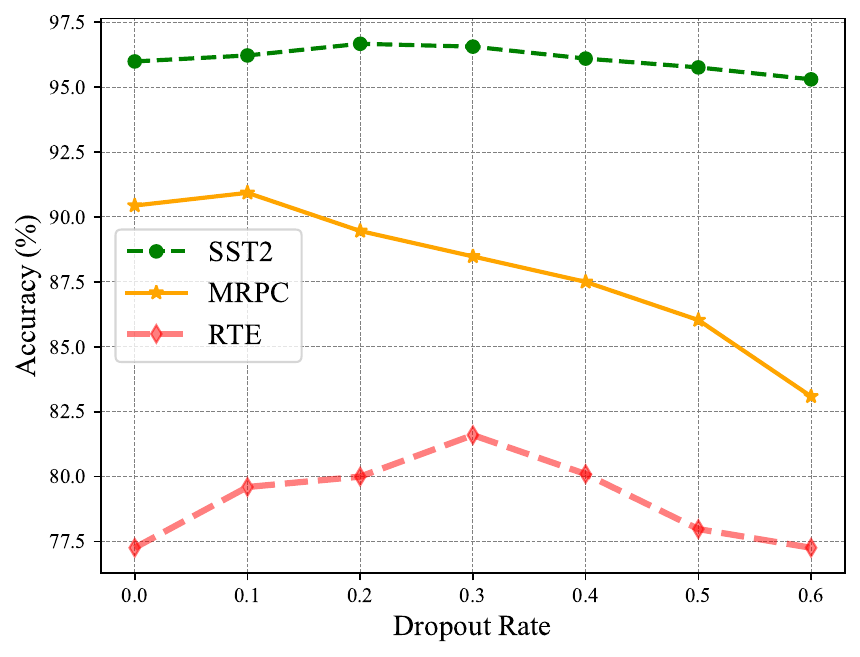}
    }    
    \subfigure[Accuracy v.s. dropout rate on LLaMA]
    {
    \label{Fig:FedLoDrop_rate_LLA}
    \includegraphics[width=0.4\textwidth]{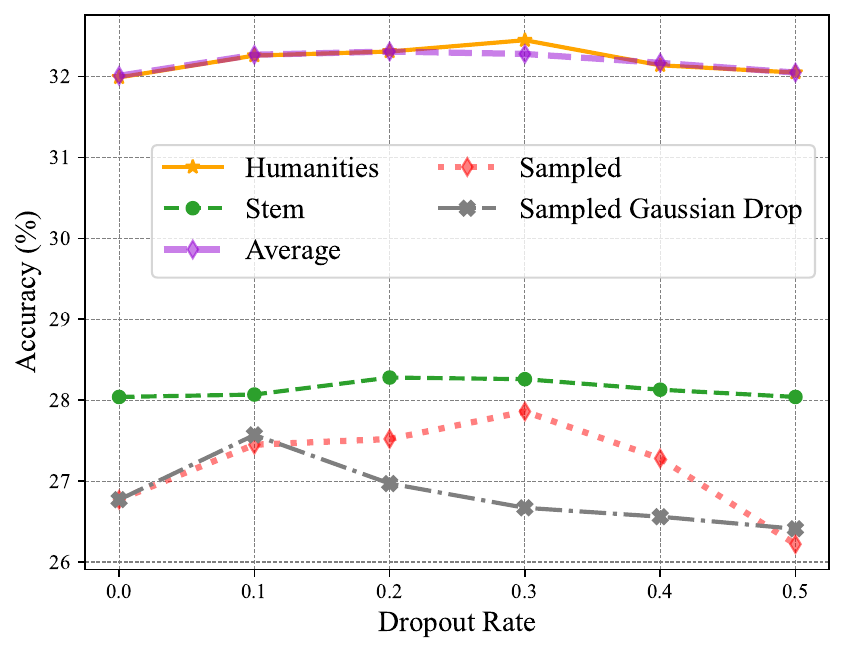}
    } 
    \caption{Effects of dropout rate on the performance.}% 
    \label{Fig:FedLoDrop_rate}
    \end{center}\vspace{-0.6cm}
\end{figure*}

\begin{figure*}[h] 
    \begin{center}
    \subfigure[Accuracy v.s. per-round latency on RoBERTa large]
    {
    \label{Fig:FedLoDrop_delay_rober}
    \includegraphics[width=0.4\textwidth]{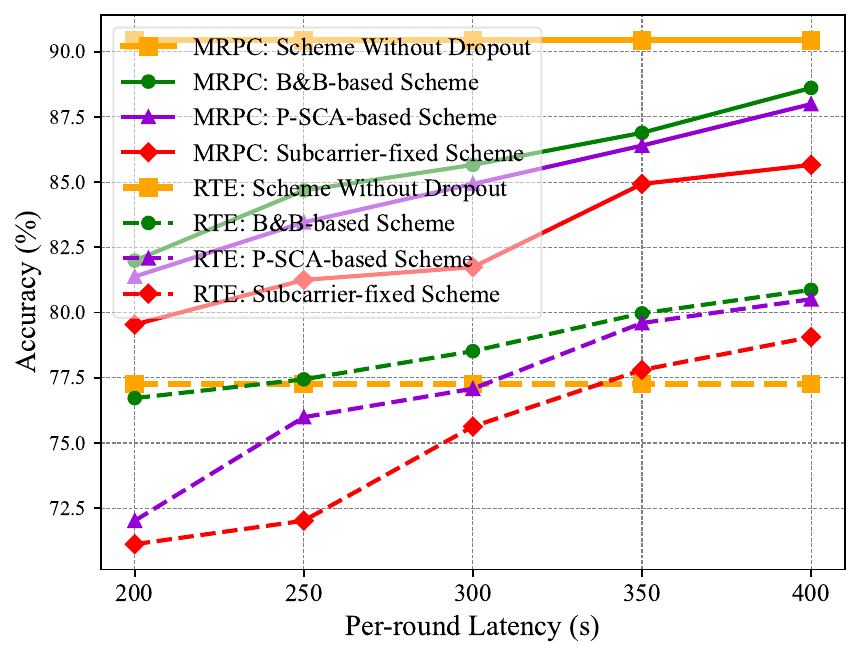}
    }    
    \subfigure[Accuracy v.s. per-round latency on LLaMA]
    {
    \label{Fig:FedLoDrop_delay_LLA}
    \includegraphics[width=0.4\textwidth]{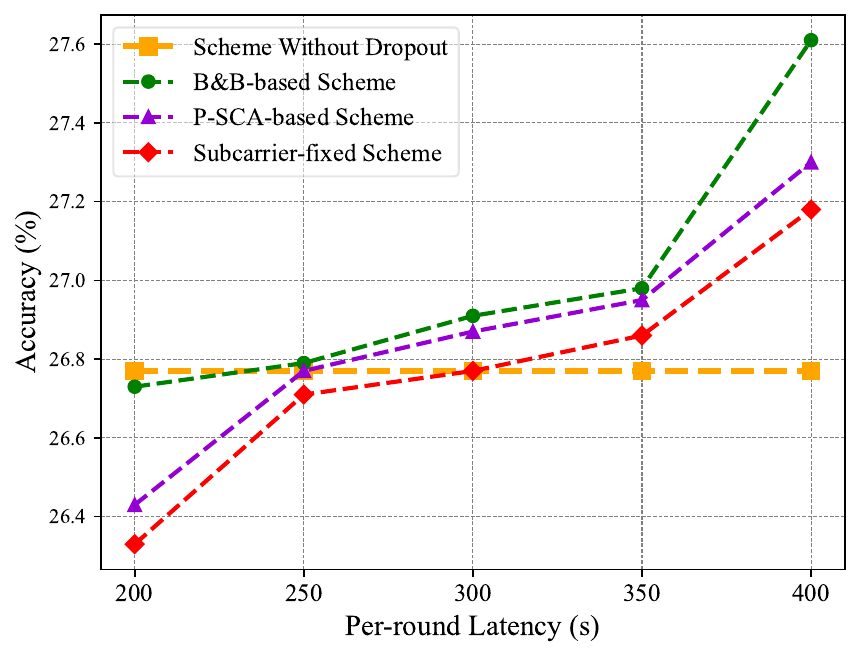}
    } 
    \caption{Effects of per-round latency on the performance.}% 
    \label{Fig:FedLoDrop_delay}
    \end{center}\vspace{-0.8cm}
\end{figure*}

\begin{figure}[htbp]
  \centering
  \includegraphics[width=0.45\textwidth]{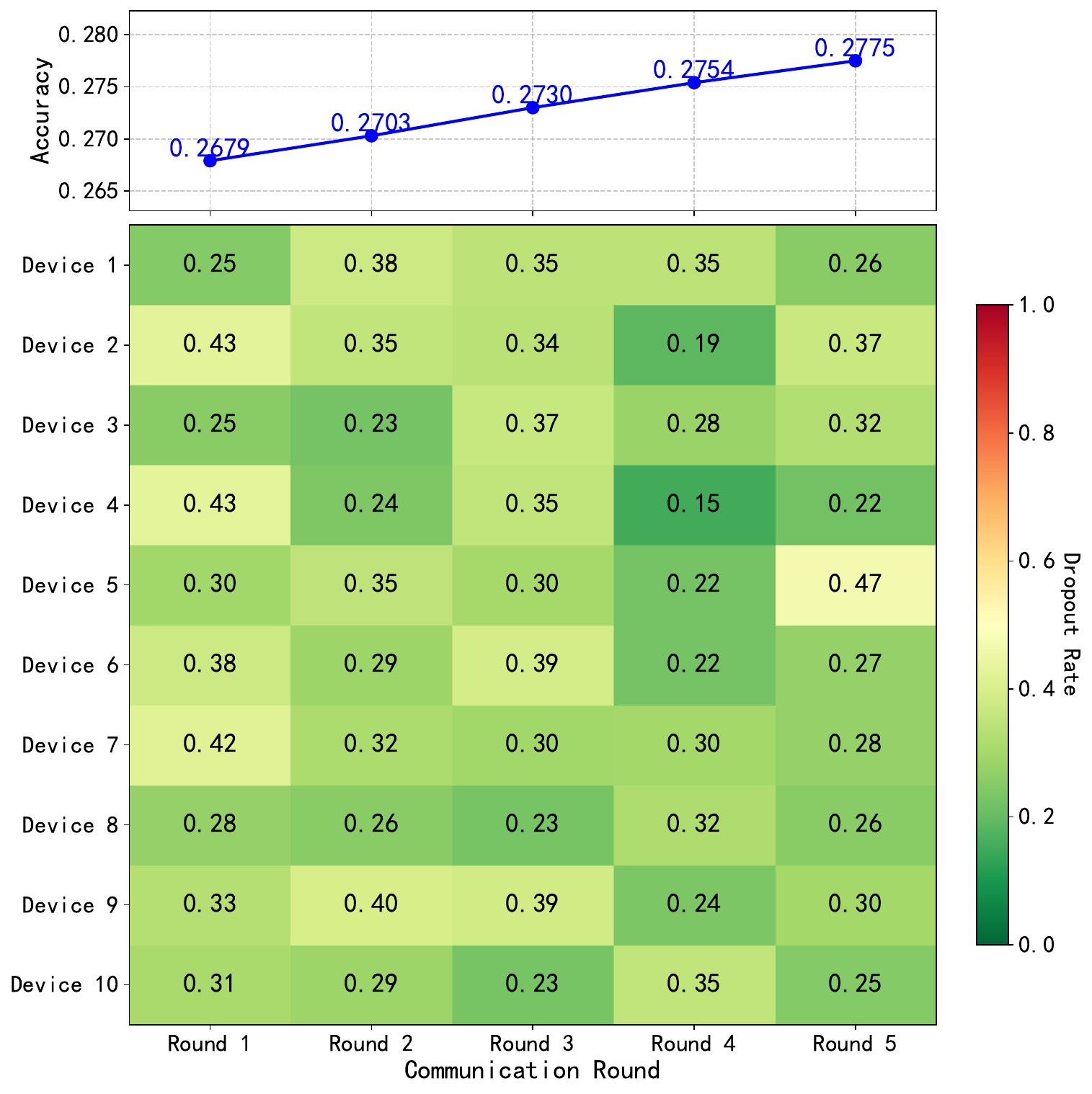}
  \caption{The evolution of the optimized dropout rates under different rounds.}
  \label{Fig:drop_change}\vspace{-0.7cm}
\end{figure}

\begin{figure*}[h] 
    \begin{center}
    \subfigure[Testing accuracy v.s. communication round]
    {
    \label{Fig:Convergence performance acc}
    \includegraphics[width=0.4\textwidth]{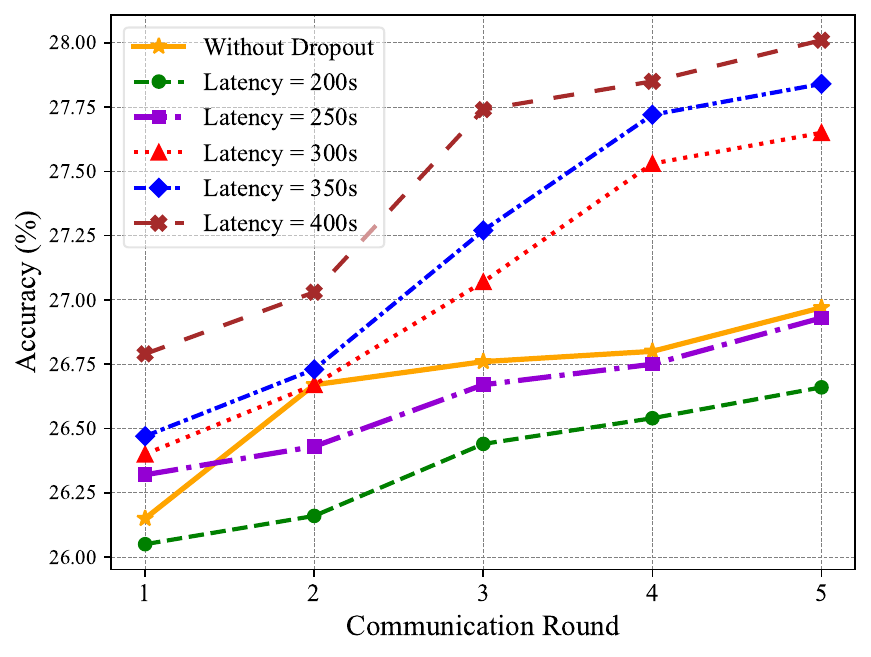}
    }    
    \subfigure[Training loss v.s. epoch]
    {
    \label{Fig:Convergence performance loss}
    \includegraphics[width=0.5\textwidth]{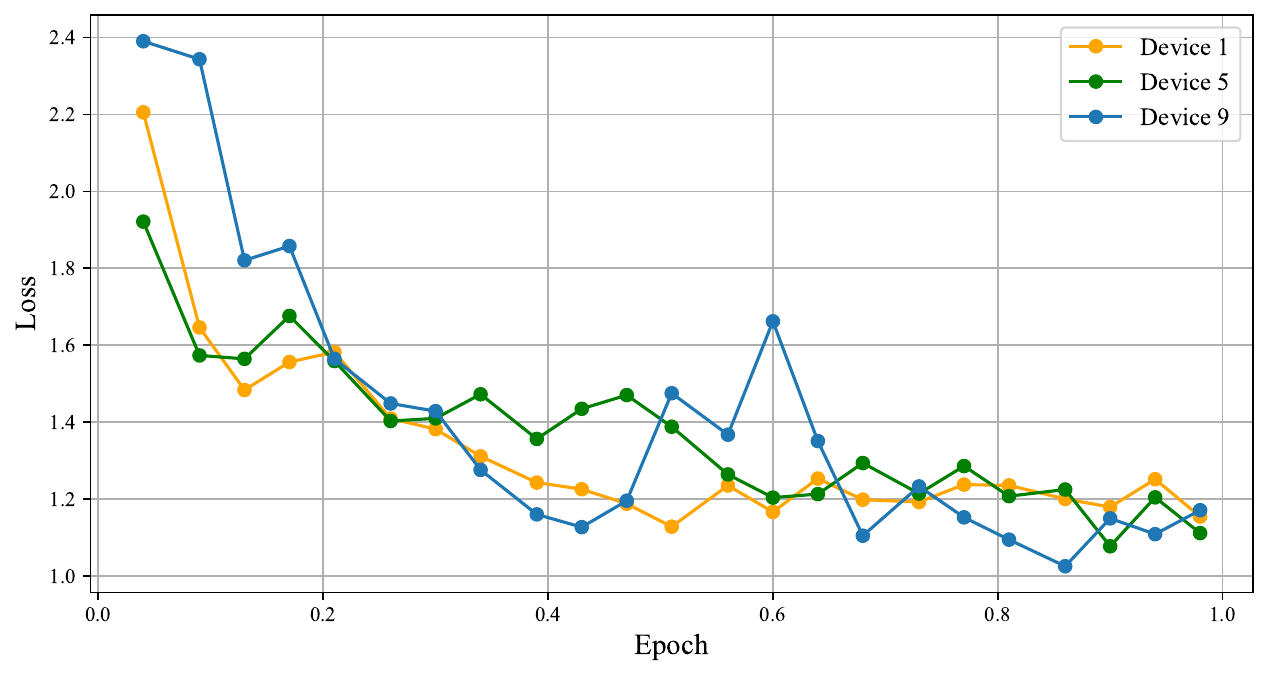}
    } 
    \caption{Convergence performance.}% 
    \label{Fig:Convergence performance}
    \end{center}\vspace{-0.6cm}
\end{figure*}

\begin{figure*}[h] 
    \begin{center} 
    \subfigure[Effects of number of devices on RoBERTa large]
    {
    \label{Fig:devicenum}
    \includegraphics[width=0.31\textwidth]{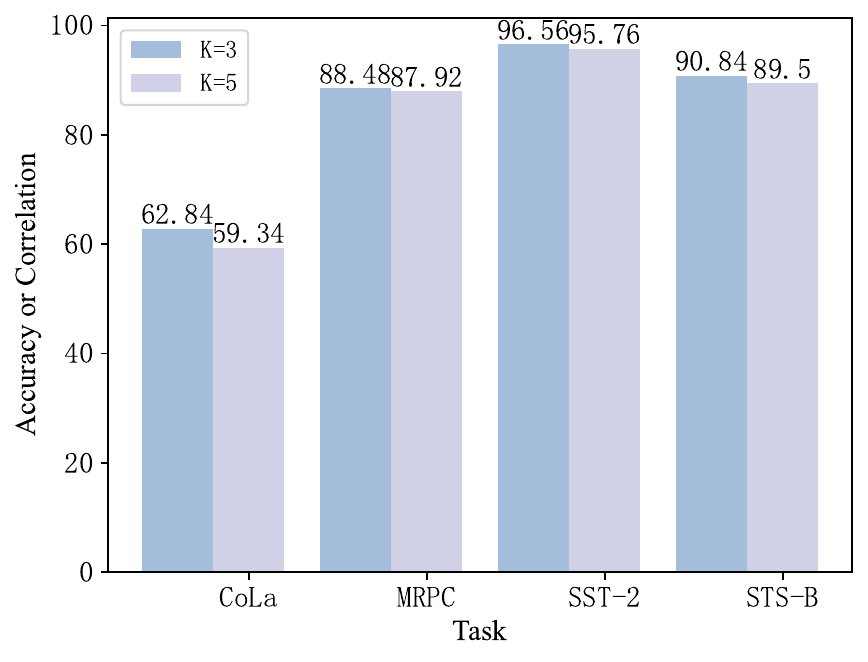}
    }
    \textcolor{black}{
    \subfigure[Effects of number of devices on LLaMA (7B)]
    {
    \label{Fig:devicenum-llama}
    \includegraphics[width=0.31\textwidth]{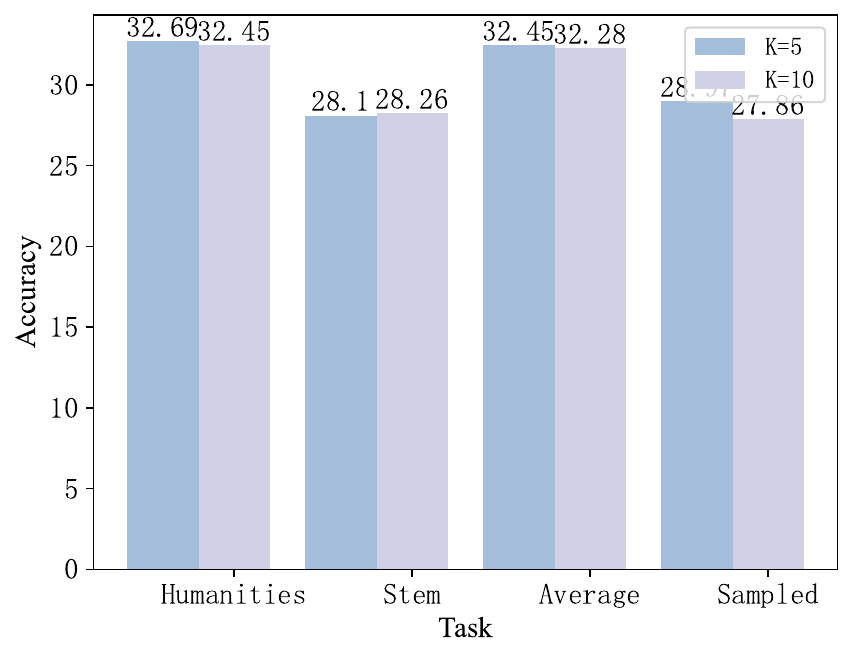}
    } 
    }
    \subfigure[\textcolor{black}{Effects of ranks}]
    {
    \label{Fig:rank}
    \includegraphics[width=0.31\textwidth]{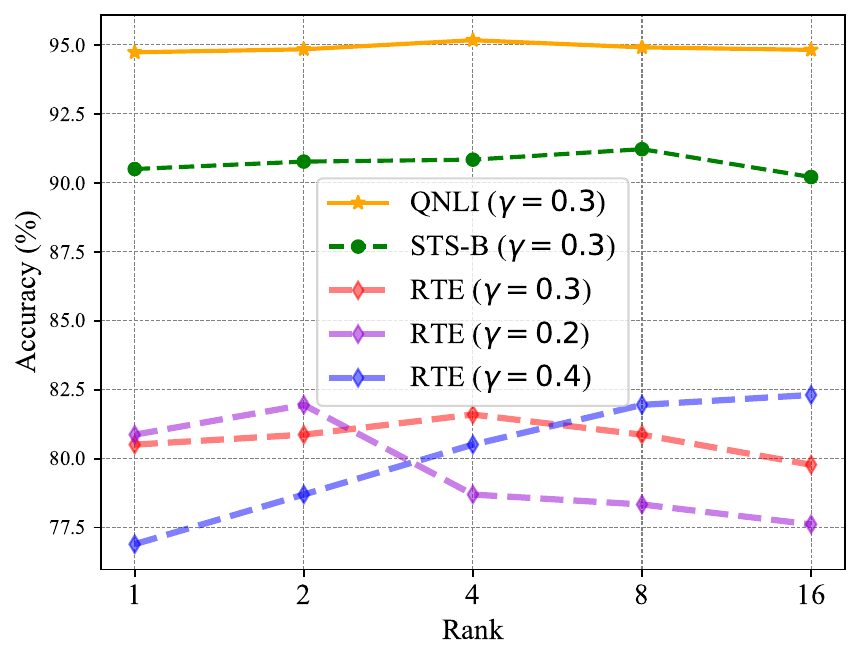}
    } 
    \textcolor{black}{\caption{Effects of parameters on the performance.}}% 
    \label{fig:para}
    \end{center}\vspace{-0.7cm}
\end{figure*}

\textcolor{black}{Fixed-rate experiments are conducted for ablation and sensitivity analysis, which establish a baseline performance landscape to compare how different levels of regularization (dropout rates) affect final model accuracy and convergence.} Table \ref{table:glue} presents the results of RoBERTa-large on the GLUE benchmark. It can be observed that models incorporating the dropout method generally achieve superior performance, suggesting that the proposed FedLoDrop framework helps LoRA-based models control overfitting and enhance generalization on downstream tasks. The framework is also applicable to FFA \cite{sun2024improving}. In contrast to FedLoDrop, when the dropout rate is set to $0.3$, the average score begins to decline. This is attributed to fewer parameters being updated under this setting, which limits model adaptability and consequently reduces performance. 

To visualize the trend, we plot the effects of dropout rate on the performance, as shown in Fig. \ref{Fig:FedLoDrop_rate}. For LLaMA, we selected 1444 samples from the dataset for a quick and comprehensive evaluation, noted as Sampled. For MMLU, among 57 disciplines, Humanities, Stem, and Average are plotted. In both figures, as the dropout rate increases, the accuracy first improves and then drops. This aligns with the theoretical derivation that a proper dropout rate would help balance the adaptation function's empirical risk minimization and complexity. A small dropout rate might fail to introduce sufficient sparsity and lead to overfitting, while an excessively large dropout rate would result in too few trainable parameters, making the adapter lose its expressive power. \textcolor{black}{This contrast underscores the necessity of finding an optimal, and potentially dynamic, dropout strategy.}

\textcolor{black}{
To evaluate sensitivity to the dropout distribution, we compare Bernoulli and Gaussian dropout. Gaussian dropout multiplies activations by a continuous random variable drawn from $\mathcal{N}(1, \sigma_{gaus}^2)$ instead of setting them to zero. The results on RoBERTa-large and LLaMA models are presented in Table \ref{table:glue} and Fig. \ref{Fig:FedLoDrop_rate_LLA}, respectively. The comparable performance of both methods demonstrates the robustness to the specific implementation of dropout. Bernoulli dropout provides more stable results across a wider range of rates, while Gaussian dropout yields slightly better performance at lower rates (e.g., 0.2 for RoBERTa-large, 0.1 for LLaMA), likely because its continuous noise injection prevents the complete silencing of neurons. In contrast, higher dropout rates lead to exploding variance, potentially destabilizing the training process.
}

\subsection{Effects of Network Resources}
To verify the superiority of the proposed two schemes, the following benchmark schemes are compared under different network resources. 
\emph{1) Scheme Without Dropout:} Sufficient resource is allocated and dropout is not considered, which is an ideal benchmark used to compare with the proposed schemes. 
\emph{2) Subcarrier-fixed Scheme:} Subcarrier is fixed, and then optimize the remaining variables the same as the proposed ones.
Fig. \ref{Fig:FedLoDrop_delay} illustrates the model accuracy v.s. different per-round latencies. \textcolor{black}{It can be observed that, under a given per-round latency threshold, both the proposed B\&B-based and the P-SCA-based algorithms outperform the subcarrier-fixed scheme when the constraint is satisfied.} The P-SCA-based algorithm is capable of obtaining a suboptimal solution with performance close to that of the optimal method. Furthermore, as the allowable per-round training delay increases, the model performance tends to improve. This is attributed to a looser constraint allowing for a lower dropout rate, which effectively reduces the noise introduced during training and thereby enhances overall performance.

\textcolor{black}{
The evolution of the optimized dropout rates for a subset of clients over the training rounds can be shown in Fig. \ref{Fig:drop_change}. Fig. \ref{Fig:drop_change} is plotted under the per-round latency of 400s. This visualization clearly shows how our algorithm dynamically adjusts the dropout rate (regularization strength) per device based on their real-time CSI and state. Simultaneously, the testing accuracy improves as the number of communication rounds increases, indicating effective learning despite adaptive parameter reduction.
}

\textcolor{black}{
In addition to reporting the final accuracy to evaluate generalization performance, we present the global testing accuracy v.s. communication round and the device training loss v.s. epoch to illustrate the convergence behavior, as shown in Fig. \ref{Fig:Convergence performance acc} and Fig. \ref{Fig:Convergence performance loss}, respectively. For clarity, we plot the training loss of three representative devices. The standard federated LoRA scheme is denoted by 'Without Dropout'. It can be observed that the convergence rate improves under a looser per-round latency constraint, which aligns with the theoretical analysis in Section~\ref{subsect:conver}.
}

\textcolor{black}{
\subsection{Effects of Resource Utilization Efficiency}
}

\textcolor{black}{
\subsubsection{Per-round Latency}
As shown in Fig. \ref{Fig:FedLoDrop_delay_rober} and Fig. \ref{Fig:FedLoDrop_delay_LLA},  given a fixed target accuracy, the proposed B\&B-based and P-SCA-based algorithms achieve the same performance with shorter per-round latency compared to the subcarrier-fixed scheme. This improvement is attributed to the reduced number of transmitted parameters. Notably, when targeting a higher accuracy (e.g., 27.2\% on LLaMA), the subcarrier-fixed scheme and the scheme without dropout fail to achieve the desired performance (27.2\% testing accuracy), highlighting the limitations of fixed resource allocation under stringent latency and accuracy requirements.
}

\textcolor{black}{
\subsubsection{Communication Overhead}
To quantify communication overhead, we compare the communicated parameters of the full model FT, FedIT, and FedLoDrop across 5 communication rounds. The results show that FedLoDrop achieves the lowest communication cost, transmitting only 7.6\% of the parameters required by full FT, while FedIT transmits 9.5\%.  By transmitting only a subset of the model parameters, FedLoDrop significantly reduces the volume of data exchanged per round compared to full FT approaches.
% \vspace{-0.6cm}
% \begin{table}[htbp]
%     \centering
%     \caption{The ratio of communicated parameter numbers to Full FT.}
%     \label{table:parameter}
%     \begin{tabular}{lcc} % l表示左对齐，c表示居中，r表示右对齐
%         \toprule
%         FT & FedIT & FedLoDrop \\
%         \midrule
%         1 & 0.095   & 0.076   \\
%         \bottomrule
%     \end{tabular}\vspace{-0.7cm}
% \end{table}
}

\subsection{Effects of Parameters and Models}

Fig. \ref{fig:para} presents the effects of parameters on the performance. From Fig. \ref{Fig:devicenum} and Fig. \ref{Fig:devicenum-llama}, one can see that in both models, a larger number of devices generally results in a lower testing accuracy due to the difficulty in adapting to one global model. \textcolor{black}{We further evaluate the sampled setting with 50 devices on LLaMA, where the accuracy decreases modestly to 26.57. As shown in Fig. \ref{Fig:rank}, with $\gamma = 0.3$, performance first improves with rank but slightly declines when the rank becomes large due to overfitting. For RTE, we further include results under dropout rates of 0.2 and 0.4. It can be observed that insufficient dropout (0.2) leads to overfitting at high ranks, while strong dropout (0.4) can initially cause underfitting but enables higher ranks to be effectively utilized.}

\textcolor{black}{
As a final stress test for FedLoDrop, we scale up to Yi-34B. Due to the high training cost, we only report the results for given dropout rates. The LoRA rank is set as $16$, with the learning rate of $3e-4$. Experiments are performed on 2 NVIDIA A100 (40 GB each) GPUs. As shown in Table \ref{table:yi34B}, FedLoDrop with 0.2 dropout rate achieves the best performance, consistent with the findings in Section \ref{subsect:fix_drop}. For brevity and due to space constraints, we omit the full results.
\vspace{-0.7cm}
\begin{table}[htbp]
    \centering
    \caption{Accuracy v.s. dropout rate on Yi-34B.}
    \label{table:yi34B}
    \begin{tabular}{lcccc} % l表示左对齐，c表示居中，r表示右对齐
        \toprule
        Dropout Rate & 0 & 0.1 & 0.2 & 0.3 \\
        \midrule
        Accuracy &  43.04  & 43.35  & \textbf{43.39}   &  42.54\\
        \bottomrule
    \end{tabular}\vspace{-0.7cm}
\end{table}
}

\section{Conclusion}
In this paper, we propose a new framework, called FedLoDrop, which effectively incorporates a parameter-efficient finetuning technique, known as LoRA, to facilitate local training.  This method reduces computational and communication overheads for local edge devices that have limited resources. We analytically demonstrate the trade-off between underfitting and overfitting by deriving mechanisms that elucidate these dynamics. The application of dropout alleviates the imbalanced update of the parameter matrix and mitigates parameter overfitting in LoRA. Specifically, PHS-based analysis reveals that increasing the dropout rate narrows the gap between empirical and generalization errors but also increases empirical error. Based on this insight, we propose two schemes, B\&B-based and P-SCA-based algorithms to minimize the generalization error bound by jointly optimizing the dropout rate and network resource allocation, thereby enhancing the learning performance. Frameworks for collaborative federated sensing and communication systems will be left for future directions. \textcolor{black}{Besides, in future work, the proposed FedLoDrop can be extended to multi-modal large language models (MLLMs) to dynamically adjust local dropout strategies based on modality importance or client-specific resource constraints.}

\vspace{-0.3cm}
\appendices
\section{Proof Of Theorem \ref{theorem:phs}}\label{proof_of_theorem}
\subsection{Proof of Theorem \ref{theorem:FedLoDrop_PHS}}

\begin{lemma} \label{Lem:regul}
    Consider the learning algorithm M optimizing the following loss function:
    \begin{equation}
        \min_{\boldsymbol{\theta}_{k,t}} \ell_{k,\lambda_t}({\boldsymbol{\theta}}_{k,t}) =  \min_{\boldsymbol{\theta}_{k,t}} \ell_{k,t}({\boldsymbol{\theta}}_{k,t}) + \lambda_t \| {\boldsymbol{\theta}}_{k,t} - \boldsymbol{\theta}_0 \|_2^2.
    \end{equation}
    If the requirements in Definition \ref{definition} are met, then it has PHS $\beta=\frac{2\eta^2}{ (\Lambda_{k,t,\text{min}} +2\lambda_t) \left|\mathcal{D}_k\right|}$, which is
    \begin{equation}
        \begin{aligned}
                    &     \mathbb{E}_{\mathcal{D}_k, j \sim U(n)}\left|\ell_{k, \lambda_t} \left( \boldsymbol{\theta}_{ \ell_{k,\lambda_t}}(\mathcal{D}_k^j),\boldsymbol{x}_j \right)- \ell_{k,\lambda_t} \left( \boldsymbol{\theta}_{ \ell_{k,\lambda}}(\mathcal{D}_k),\boldsymbol{x}_j \right) \right| \\
                    & \quad\quad \leq \frac{2\eta^2}{ (\Lambda_{k,t,\text{min}} +2\lambda_t) \left|\mathcal{D}_k\right|}.
        \end{aligned}
    \end{equation}
\end{lemma}

\begin{proof}
   Denote $\boldsymbol{\theta}_{ \ell_{k,\lambda_t}}(\mathcal{D}_k)$ as $\Hat{\boldsymbol{\theta}}_{k,t}$, and $\Delta \hat{\boldsymbol{\theta}}_{k,t} = \hat{\boldsymbol{\theta}}_{k,t} - \boldsymbol{\theta}$. Consider the second-order Taylor expansion of $\ell_{k,\lambda_t}$ at local optimal $\hat{\boldsymbol{\theta}}_{k,t}$, $\nabla_{\ell_{k, \lambda_t}}(\hat{\boldsymbol{\theta}}_{k,t})=0$. Thus, for $\forall v$ close to $\hat{\boldsymbol{\theta}}_{k,t}$, 
   
% \begin{figure*}[htbp]
%     \centering
    \begin{equation} \label{equ:taylor}
    \begin{aligned}
        \ell_{k, \lambda_t}(v) 
        &= \ell_{k,\lambda_t}(\hat{\boldsymbol{\theta}}_{k,t}) + \frac{1}{2}(v-\hat{\boldsymbol{\theta}}_{k,t})^{\top}\nabla^2{\ell_{k,\lambda_t}}(\hat{\boldsymbol{\theta}}_{k,t})(v-\hat{\boldsymbol{\theta}}_{k,t}).
    \end{aligned}
\end{equation}
% \vspace*{8pt}
% \hrulefill
% \end{figure*}

Then, expand $\nabla^2{\ell_{k, \lambda_t}}(\hat{\boldsymbol{\theta}}_k)$ in \eqref{equ:taylor}, 
\begin{equation} 
    \begin{aligned}
        & \nabla^2{\ell_{k, \lambda_t}}(\hat{\boldsymbol{\theta}}_k)
        = \nabla^2_{\hat{\boldsymbol{\theta}}_{k,t}}(\ell_k(\hat{\boldsymbol{\theta}}_{k,t})+ \lambda_t \| \hat{\boldsymbol{\theta}}_{k,t} - \boldsymbol{\theta}_0 \|_2^2)\\
        & = \nabla^2 \ell_k(\hat{\boldsymbol{\theta}}_{k,t}) + 2\lambda_t I = U_{k,t} \text{diag}(\Lambda_{k,t})U_{k,t}^{-1} + 2\lambda_t I\\
        & = U_{k,t} (\text{diag}(\Lambda_{k,t})+ 2\lambda_t I ) U_{k,t}^{-1} \\
        & = U_{k,t} \text{diag}(\sqrt{\Lambda_{k,t,1}+2\lambda_t}, \cdots, \sqrt{\Lambda_{k,t,d}+2\lambda_t}) U_{k,t}^{-1} U_{k,t} \\
        & \quad \text{diag}(\sqrt{\Lambda_{k,t,1}+2\lambda_t}, \cdots, \sqrt{\Lambda_{k,t,d}+2\lambda_t}) U_{k,t}^{-1}.
    \end{aligned}
\end{equation}
Take this back to \eqref{equ:taylor},
\begin{equation} \label{equ:basic1}
    \begin{aligned}
        & \ell_{k, \lambda_t}(v) 
        - \ell_{k, \lambda_t}(\hat{\boldsymbol{\theta}}_{k,t}) \\ 
        & = \frac{1}{2} \| (U_{k,t} \text{diag}(\sqrt{\Lambda_{k,t,1}+2\lambda_t}, \cdots, \sqrt{\Lambda_{k,t,d}+2\lambda_t}) U_{k,t}^{-1}) \\
         & \quad (v-\hat{\boldsymbol{\theta}}_{k,t}) \|_2^2 \quad \geq \frac{1}{2} (\Lambda_{k,t,\text{min}} +2\lambda_t) \|  v-\hat{\boldsymbol{\theta}}_{k,t} \|_2^2.
    \end{aligned}
\end{equation}
According to the definition of $\ell_{k,\lambda_t}(\hat{\boldsymbol{\theta}}_{k,t})$, for $\forall u,v$ close to $\hat{\boldsymbol{\theta}}_{k,t}$, \eqref{equ:basic2} is got.
\begin{figure*}[htbp]
    \centering
    \begin{equation}\label{equ:basic2}
        \begin{aligned}
            & \ell_{k, \lambda_t}(u) 
        - \ell_{k, \lambda_t}(v) 
        = \left(\frac{1}{\left|\mathcal{D}_k\right|} \sum_{x_i\in\mathcal{D}_k}\ell_{k} \left( u,\boldsymbol{x}_i \right)+ \lambda_t\|u-{\boldsymbol{\theta}}_0 \|_2^2 \right)- \left(\frac{1}{\left|\mathcal{D}_k\right|} \sum_{x_i\in\mathcal{D}_k}\ell_{k} \left( v,\boldsymbol{x}_i \right)+ \lambda_t\|v-{\boldsymbol{\theta}}_0 \|_2^2 \right)\\
        % &= \left(\frac{1}{\left|\mathcal{D}_k\right|} \sum_{\boldsymbol{x}_i\in\mathcal{D}_k, i\neq j}\ell_{k} \left( u,\boldsymbol{x}_i \right)+ \lambda_t\|u-{\boldsymbol{\theta}}_0 \|_2^2 \right) -\left(\frac{1}{\left|\mathcal{D}_k\right|} \sum_{\boldsymbol{x}_i\in\mathcal{D}_k, i\neq j}\ell_{k} \left(\boldsymbol{x}_i; v \right)+ \lambda_t\|v-{\boldsymbol{\theta}}_0 \|_2^2 \right) + \frac{\ell_{k} \left(u,\boldsymbol{x}_j \right)- \ell_{k} \left(v,\boldsymbol{x}_j \right)}{\left|\mathcal{D}_k\right|}\\
        &= \left(1-\frac{1}{\left|\mathcal{D}_k\right|}\right) \left(\frac{1}{\left|\mathcal{D}_k\right|-1} \sum_{ i\neq j}\ell_{k} \left( u,\boldsymbol{x}_i \right)+ \lambda_t\|u-{\boldsymbol{\theta}}_0 \|_2^2 \right) - \left(1-\frac{1}{\left|\mathcal{D}_k\right|}\right) \left(\frac{1}{\left|\mathcal{D}_k\right|-1} \sum_{ i\neq j}\ell_k \left( v,\boldsymbol{x}_i \right)+ \lambda_t\|v-{\boldsymbol{\theta}}_0 \|_2^2 \right)\\
        & \quad\quad + \dfrac{\lambda_t (\|u-{\boldsymbol{\theta}}_0 \|_2^2-\|v-{\boldsymbol{\theta}}_0 \|_2^2)}{\left|\mathcal{D}_k\right|}  + \frac{\ell_k \left( u,\boldsymbol{x}_j \right)- \ell_{k} \left( v, \boldsymbol{x}_j \right)}{\left|\mathcal{D}_k\right|}\\
        & = \left(1-\frac{1}{\left|\mathcal{D}_k\right|}\right) \left[ \left(\frac{1}{\left|\mathcal{D}_k\right|-1} \sum_{ i\neq j}\ell_{k} \left( u,\boldsymbol{x}_i \right)+ \lambda_t\|u-{\boldsymbol{\theta}}_0 \|_2^2 \right) - \left(\frac{1}{\left|\mathcal{D}_k\right|-1} \sum_{ i\neq j}\ell_k \left( v,\boldsymbol{x}_i \right)+ \lambda_t\|v-{\boldsymbol{\theta}}_0 \|_2^2 \right) \right]\\
        & \quad\quad + \dfrac{\ell_{k, \lambda_t} \left( u,\boldsymbol{x}_j \right)- \ell_{k, \lambda_t} \left( v,\boldsymbol{x}_j \right)}{\left|\mathcal{D}_k\right|}.
        \end{aligned}
\end{equation}
\vspace*{-20pt}
\hrulefill
\end{figure*}
Take $u=\boldsymbol{\theta}_{ \ell_{k,\lambda_t}}(\mathcal{D}_k^j)$ and $v=\boldsymbol{\theta}_{ \ell_{k,\lambda_t}}(\mathcal{D}_k)$,
\begin{equation}\label{equ:basic3}
    \begin{aligned}
            & \ell_{k, \lambda_t}(\boldsymbol{\theta}_{ \ell_{k,\lambda_t}}(\mathcal{D}_k^j)) - \ell_{k, \lambda}(\boldsymbol{\theta}_{ \ell_{k,\lambda_t}}(\mathcal{D}_k)) \\
            & \leq \frac{\ell_{k, \lambda_t} \left( \boldsymbol{\theta}_{ \ell_{k, \lambda_t}}(\mathcal{D}_k^j),\boldsymbol{x}_j \right)- \ell_{k, \lambda_t} \left( \boldsymbol{\theta}_{ \ell_{k, \lambda_t}}(\mathcal{D}_k),\boldsymbol{x}_j \right)}{\left|\mathcal{D}_k\right|},
    \end{aligned}
\end{equation}
take \eqref{equ:basic1} into \eqref{equ:basic3}, 
\begin{equation}\label{equ:basic4}
    \begin{aligned}
            & \frac{1}{2} (\Lambda_{k,t,\text{min}} +2\lambda_t) \|  \boldsymbol{\theta}_{ \ell_{k, \lambda_t}}(\mathcal{D}_k^j) -\boldsymbol{\theta}_{ \ell_{k, \lambda_t}}(\mathcal{D}_k) \|_2^2 \\
            & \leq \frac{\ell_{k, \lambda_t} \left( \boldsymbol{\theta}_{ \ell_{k,\lambda_t}}(\mathcal{D}_k^j),\boldsymbol{x}_j \right)- \ell_{k, \lambda_t} \left( \boldsymbol{\theta}_{ \ell_{k, \lambda_t}}(\mathcal{D}_k),\boldsymbol{x}_j \right)}{\left|\mathcal{D}_k\right|}.
    \end{aligned}
\end{equation}
Because the loss function is $\eta$-Lipschitz, thus,
\begin{equation}\label{equ:basic5}
    \begin{aligned}
         & \frac{ \left|\ell_{k, \lambda_t} \left( \boldsymbol{\theta}_{ \ell_{k, \lambda_t}}(\mathcal{D}_k^j),\boldsymbol{x}_j \right)- \ell_{k,\lambda_t} \left( \boldsymbol{\theta}_{ \ell_{k, \lambda_t}}(\mathcal{D}_k),\boldsymbol{x}_j \right) \right|}{\left|\mathcal{D}_k\right|}  \\ 
         & \leq  \frac{\eta \|\boldsymbol{\theta}_{ \ell_{k, \lambda_t}}(\mathcal{D}_k^j)-\boldsymbol{\theta}_{ \ell_{k, \lambda_t}}(\mathcal{D}_k) \|}{\left|\mathcal{D}_k\right|}.
    \end{aligned}
\end{equation}
Take \eqref{equ:basic5} into \eqref{equ:basic4}, and basic calculation,
\begin{equation}\label{equ:basic6}
    \|\boldsymbol{\theta}_{ \ell_{k,\lambda}}(\mathcal{D}_k^j)-\boldsymbol{\theta}_{ \ell_{k, \lambda}}(\mathcal{D}_k) \| \leq \frac{2\eta}{ (\Lambda_{k,\text{min}} +2\lambda) \left|\mathcal{D}_k\right|}.
\end{equation}
Plug \eqref{equ:basic6} into \eqref{equ:basic5},
% \begin{equation}
%     \begin{aligned}
%             & \left|\ell_{k,\lambda_t} \left( \boldsymbol{\theta}_{ \ell_{k,\lambda_t}}(\mathcal{D}_k^j),\boldsymbol{x}_j \right)- \ell_{\lambda_t} \left( \boldsymbol{\theta}_{ \ell_{k,\lambda_t}}(\mathcal{D}_k),\boldsymbol{x}_j \right) \right| \\
%             & \leq \frac{2\eta^2}{ (\Lambda_{k,t,\text{min}} +2\lambda_t) \left|\mathcal{D}_k\right|}.
%     \end{aligned}
% \end{equation}
and as this holds for any $j$ and $\mathcal{D}_k$, the proof of Lemma \ref{Lem:regul} is finished.
% \begin{equation}
%     \mathbb{E}_{\mathcal{D}_k, j \sim U(n)}\left|\ell_{k, \lambda_t} \left(x_j; \boldsymbol{\theta}_{ \ell_{k,\lambda_t}}(\mathcal{D}_k^j) \right)- \ell_{k,\lambda_t} \left(x_j; \boldsymbol{\theta}_{ \ell_{k,\lambda_t}}(\mathcal{D}_k) \right) \right| \leq \frac{2\eta^2}{ (\Lambda_{k,t,\text{min}} +2\lambda_t) \left|\mathcal{D}_k\right|}.
% \end{equation}
\end{proof}

\subsection{PHS upper bound of FedLoDrop} \label{proof_of_theorem_local_drop}
Consider loss function with sparsity regularization,
\begin{equation}
    \begin{aligned}
           & \ell_{k,\lambda_t} (\boldsymbol{\theta}_{k,t})
           \\
            &= \ell_{k,t}\left(\boldsymbol{\theta}_k\right) + \lambda_t \mathbb{E}_{\boldsymbol{d}_{k,t} \sim \text{Bern}(2\gamma_{k,t}-\gamma_{k,t}^2)} \sum_j d_{k,t,j}^2 (\boldsymbol{\theta}_{k,t,j}-\boldsymbol{\theta}_{0,j})^2\\
            & = \ell_{k,t}\left(\boldsymbol{\theta}_k\right) + \lambda_t \sum_j (\boldsymbol{\theta}_{k,t,j}-\boldsymbol{\theta}_{0,j})^2 \mathbb{E}_{\boldsymbol{d}_{k,t,j} \sim \text{Bern}(2\gamma_{k,t}-\gamma_{k,t}^2)} d_{k,j}^2\\
            & = \ell_{k,t}\left(\boldsymbol{\theta}_k\right) + \lambda_t \sum_j (\boldsymbol{\theta}_{k,t,j}-\boldsymbol{\theta}_{0,j})^2 (2\gamma_{k,t}-\gamma_{k,t}^2)\\    
            & = \ell_{k,t}\left(\boldsymbol{\theta}_k\right) + \lambda_t (2\gamma_{k,t}-\gamma_{k,t}^2) \|\boldsymbol{\theta}_{k,t}-\boldsymbol{\theta}_0 \|^2.
    \end{aligned}
\end{equation}
Integrating the result above to Lemma \ref{Lem:regul} and substituting the regularization coefficient finalizes the proof.

\vspace{-0.3cm}
\section{Proof Of Lemma \ref{lemma_layer_error}}\label{proof_of_lemma_layer}
Denote $\hat{\boldsymbol{J}}_{u,k,A,t-1}=\Delta{\boldsymbol{A}}_{u,k,t-1}-\Delta\hat{\boldsymbol{A}}_{u,k,t-1}$, $\hat{\boldsymbol{J}}_{u,k,B,t-1}=\Delta{\boldsymbol{B}}_{u,k,t-1}-\Delta\hat{\boldsymbol{B}}_{u,k,t-1}$, thus 

\begin{equation}\label{equ:errorbound1}
    \begin{aligned}
    & \mathbb{E}\left[\| \boldsymbol{J}_{u,t-1} \|_F^2 \right] \\
    & = \mathbb{E}\left[\| \sum_{k=1}^K  \frac{\left|\mathcal{D}_k\right|}{|\mathcal{D}|} ( \boldsymbol{G}_{u,k,t-1}-\hat{\boldsymbol{G}}_{u,k,t-1}) \|_F^2 \right]\\
    & = \mathbb{E}\left[\|\sum_{k=1}^K  \frac{\left|\mathcal{D}_k\right|}{|\mathcal{D}|} (\boldsymbol{B}_{u,t-1} \hat{\boldsymbol{J}}_{u,k,A,t-1} + \hat{\boldsymbol{J}}_{u,k,B,t-1}\boldsymbol{A}_{u,t-1} )\|_F^2 \right]\\
    % & \stackrel{(a)} \leq 2 \mathbb{E}\left[ \|\sum_{k=1}^K  \frac{\left|\mathcal{D}_k\right|}{|\mathcal{D}|} \boldsymbol{B}_{u,t-1} \hat{\boldsymbol{J}}_{u,k,A,t-1} \|_F^2\right] \\
    % & \quad\quad + 2 \mathbb{E}\left[ \|\sum_{k=1}^K  \frac{\left|\mathcal{D}_k\right|}{|\mathcal{D}|} \boldsymbol{A}_{u,t-1} \hat{\boldsymbol{J}}_{u,k,B,t-1} \|_F^2 \right]\\
    & \stackrel{(a)} \leq 2 \sum_{k=1}^K  \frac{\left|\mathcal{D}_k\right|}{|\mathcal{D}|} \mathbb{E}\left[ \| \boldsymbol{B}_{u,t-1} \hat{\boldsymbol{J}}_{u,k,A,t-1} \|_F^2\right] \\
    & \quad\quad + 2 \sum_{k=1}^K  \frac{\left|\mathcal{D}_k\right|}{|\mathcal{D}|} \mathbb{E}\left[  \boldsymbol{A}_{u,t-1} \hat{\boldsymbol{J}}_{u,k,B,t-1} \|_F^2 \right],
    \end{aligned}
\end{equation}
where (a) comes from the convexity of F-norm. The Taylor expansion is adopted to approximate the weights of neural networks, i.e., $\Delta \hat{\boldsymbol{A}}_{u,k,t-1} 
      = 
     \Delta {\boldsymbol{A}}_{u,k,t-1} + O\left( \hat{\boldsymbol{A}}_{u,k,t-1}-{\boldsymbol{A}}_{u,k,t-1}\right) +\boldsymbol{H} ({\boldsymbol{A}}_{u,k,t-1})\left(\hat{\boldsymbol{A}}_{u,k,t-1}-{\boldsymbol{A}}_{u,k,t-1}\right)$.
Under Assumption \ref{assumption3}, the last higher order on the right side can be ignored. Thus, 
% the expectation can be bounded as
% % \begin{equation}
% %     \hat{\boldsymbol{J}}_{k,A,t-1} = -\boldsymbol{H} ({\boldsymbol{A}}_{k,t-1})\left(\hat{\boldsymbol{A}}_{k,t-1}-{\boldsymbol{A}}_{k,t-1}\right)
% % \end{equation}

\begin{equation}\label{equ:errorbound2}
    \begin{aligned}
            & \mathbb{E}\left[\| \hat{\boldsymbol{J}}_{u,k,A,t-1} \|_F^2 \right] \\
            & = \mathbb{E}\left[\| -\boldsymbol{H} ({\boldsymbol{A}}_{u,k,t-1})\left(\hat{\boldsymbol{A}}_{u,k,t-1}-{\boldsymbol{A}}_{u,k,t-1}\right) \|_F^2 \right]\\
            % & \leq \mathbb{E}\left[\| -\boldsymbol{H} ({\boldsymbol{A}}_{u,k,t-1})\|_F^2 \cdot \| \left(\hat{\boldsymbol{A}}_{u,k,t-1}-{\boldsymbol{A}}_{u,k,t-1}\right) \|_F^2 \right]\\
            & \leq \mathcal{H}^2 \cdot \mathbb{E} \left[ \| {\boldsymbol{A}}_{u,k,t-1}  \left(\mathrm{diag}(\boldsymbol{m}_{A,k,t-1})- \boldsymbol{I}\right) \|_F^2 \right] \\
            & \leq \mathcal{H}^2 \mathcal{G}^2 \mathbb{E} \left[\sum_{q=1}^{n_2} ({m}_{A,k,t-1,q}-1)^2 \right] = \mathcal{H}^2 \mathcal{G}^2 n_2 \gamma_{k,t}.
    \end{aligned}
\end{equation}
Take \eqref{equ:errorbound2} back into \eqref{equ:errorbound1}, and minor changes to $\mathbb{E}\left[\| \hat{\boldsymbol{J}}_{u,k,B,t-1} \|_F^2 \right]$, Lemma \ref{lemma_layer_error} has been proven.

\vspace{-0.3cm}
\section{Proof Of Theorem \ref{theorem:loss_descent}}\label{proof_of_theorem_conver}
According to Assumption \ref{assumption1}, set $\alpha=\frac{1}{\eta}$, and take expectation on both sides,
% \begin{equation}\label{equ:proof_whole_model}
%         \begin{aligned}
%                     & \mathbb{E} [\mathcal{L}\left(\Delta\boldsymbol{\theta}_{t}\right) ] - \mathcal{L}\left(\Delta \boldsymbol{\theta}_{t-1}\right) \\
%                     & \leq - \frac{1}{2\eta}\|\nabla\mathcal{L}(\Delta\boldsymbol{\theta}_{t-1}) \|_F^2 + \frac{1}{2\eta}\mathbb{E}\left[\| \boldsymbol{J}_{t-1} \|_F^2 \right]\\
%                     & \leq - \frac{\mu\rho}{\eta} + \frac{1}{2\eta}\mathbb{E}\left[\| \boldsymbol{J}_{t-1} \|_F^2 \right].
%         \end{aligned}
%     \end{equation}
\begingroup
\renewcommand{\arraystretch}{0.85}
\begin{equation}\label{equ:proof_whole_model}
    \begin{aligned}
        & \mathbb{E} [\mathcal{L}\left(\Delta\boldsymbol{\theta}_{t}\right) ] - \mathcal{L}\left(\Delta \boldsymbol{\theta}_{t-1}\right) \\
        & \leq - \frac{1}{2\eta}\|\nabla\mathcal{L}(\Delta\boldsymbol{\theta}_{t-1}) \|_F^2 + \frac{1}{2\eta}\mathbb{E}\left[\| \boldsymbol{J}_{t-1} \|_F^2 \right]\\
        & \leq - \frac{\mu\rho}{\eta} + \frac{1}{2\eta}\mathbb{E}\left[\| \boldsymbol{J}_{t-1} \|_F^2 \right].
    \end{aligned}
\end{equation}
\endgroup
As ${{\boldsymbol{J}}}_{t-1} \triangleq 
 \{ {{\boldsymbol{J}}}_{u,t-1} \}_{u=1}^{U^{\prime}} $, it can be further written as
 \begin{equation}
    \begin{aligned}
             &\mathbb{E}\left[\| \boldsymbol{J}_{t-1} \|_F^2 \right]  = 
             \mathbb{E}\left[ \left \|
             \begin{pmatrix}
            \boldsymbol{J}_{1,t-1} \\
            \cdots\\
            \boldsymbol{J}_{u,t-1} \\
            \cdots\\
            \boldsymbol{J}_{U^{\prime},t-1}
            \end{pmatrix}
            \right \|_F^2
             \right] \\
            %  & = 
            %  \mathbb{E}\left[ \left \|
            %  \begin{pmatrix}
            % \boldsymbol{J}_{1,t-1} \\
            % \cdots\\
            % \boldsymbol{0} \\
            % \cdots\\
            % \boldsymbol{0}
            % \end{pmatrix} + \cdots +
            % \begin{pmatrix}
            % \boldsymbol{0} \\
            % \cdots\\
            % \boldsymbol{J}_{u,t-1} \\
            % \cdots\\
            % \boldsymbol{0}
            % \end{pmatrix} + \cdots 
            % % +
            % % \begin{pmatrix}
            % % \boldsymbol{0} \\
            % % \cdots\\
            % % \boldsymbol{0} \\
            % % \cdots\\
            % % \boldsymbol{J}_{U^{\prime},t-1}
            % % \end{pmatrix}
            % \right \|_F^2
            %  \right] \\
            %  & = 
            %  \mathbb{E}\left[ \left \|
            %  \begin{pmatrix}
            % \boldsymbol{J}_{1,t-1} \\
            % \cdots\\
            % \boldsymbol{0} \\
            % \cdots\\
            % \boldsymbol{0}
            % \end{pmatrix} \right \|_F^2 + \cdots + 
            % \left \| \begin{pmatrix}
            % \boldsymbol{0} \\
            % \cdots\\
            % \boldsymbol{J}_{u,t-1} \\
            % \cdots\\
            % \boldsymbol{0}
            % \end{pmatrix} \right \|_F^2 + \cdots 
            % % + 
            % % \left \| \begin{pmatrix}
            % % \boldsymbol{0} \\
            % % \cdots\\
            % % \boldsymbol{0} \\
            % % \cdots\\
            % % \boldsymbol{J}_{U^{\prime},t-1}
            % % \end{pmatrix}
            % % \right \|_F^2
            %  \right] \\
             % & =\mathbb{E}\left[ \| \boldsymbol{J}_{1,t-1} \|_F^2 + \cdots + \| \boldsymbol{J}_{u,t-1} \|_F^2  
             % + \cdots + \| \boldsymbol{J}_{U^{\prime},t-1} \|_F^2 \right] \\
             & = \sum_{u=1}^{U^{\prime}}\mathbb{E}\left[\| \boldsymbol{J}_{u,t-1} \|_F^2 \right]
             = {U^{\prime}}\mathbb{E}\left[\| \boldsymbol{J}_{u,t-1} \|_F^2 \right].
    \end{aligned}
 \end{equation}
Taking this back to \eqref{equ:proof_whole_model}, and combining it with Lemma \ref{lemma_layer_error}, the proof is completed.

\vspace{-0.3cm}
 \section{Proof Of Lemma \ref{lemma_convex}}\label{proof_of_lemma_convex}
All the constraints and objective function can be derived by substituting the variable transformation in \eqref{equ:aul_val}. The first term in the objective function is convex. For the second one, considering two dimensions, e.g., $\gamma_1$, $\gamma_2$, its Hessian matrix is 
$$
    \left[
        \begin{array}{clr}
        o_{11} 
        & o_{12}\\
        o_{21}  & o_{22}
        \end{array}
    \right]
$$
where, $o_{11}=\frac{4\gamma_1^2 }{(a_1-\Tilde{\gamma_1}^2)^3\sqrt{\frac{1}{a_1-\Tilde{\gamma_1}^2} + \frac{1}{a_2-\Tilde{\gamma_2}^2}}}-\frac{\gamma_1^2}{(a_1-\Tilde{\gamma_1}^2)^4(\frac{1}{a_1-\Tilde{\gamma_1}^2} + \frac{1}{a_2-\Tilde{\gamma_2}^2})^{3/2}} + \frac{1 }{(a_1-\Tilde{\gamma_1}^2)^2\sqrt{\frac{1}{a_1-\Tilde{\gamma_1}^2} + \frac{1}{a_2-\Tilde{\gamma_2}^2}}}$, $o_{12}=o_{21}=-\frac{\gamma_1\gamma_2}{(a_1-\Tilde{\gamma_1}^2)^2(a_2-\Tilde{\gamma_2}^2)^2 (\frac{1}{a_1-\Tilde{\gamma_1}^2} + \frac{1}{a_2-\Tilde{\gamma_2}^2})^{3/2} }$, $o_{22}=\frac{4\gamma_2^2 }{(a_2-\Tilde{\gamma_2}^2)^3\sqrt{\frac{1}{a_1-\Tilde{\gamma_1}^2} + \frac{1}{a_2-\Tilde{\gamma_2}^2}}}-\frac{\gamma_2^2}{(a_2-\Tilde{\gamma_2}^2)^4(\frac{1}{a_1-\Tilde{\gamma_1}^2} + \frac{1}{a_2-\Tilde{\gamma_2}^2})^{3/2}} + \frac{1 }{(a_2-\Tilde{\gamma_2}^2)^2\sqrt{\frac{1}{a_1-\Tilde{\gamma_1}^2} + \frac{1}{a_2-\Tilde{\gamma_2}^2}}}$, and is non-negative. The higher dimensions can be generalized with inductions and definitions. Therefore, the overall objective function is convex. Besides, all constraints are linear relations, and thus are convex. 
% % Thereby, the feasible region of these constraints are convex sets. 
% Thus, Problem $\mathscr P_2$ is convex.

\vspace{-6mm}
\bibliographystyle{IEEEtran}
\bibliography{refs}

\end{document}